\documentclass[a4paper, 12pt]{article}

\usepackage{booktabs}   
\usepackage{subcaption} 

\usepackage{amsmath}
\usepackage{amssymb}
\usepackage{amsthm}
\usepackage{mathtools}
\usepackage{url}
\usepackage[all]{xy}
\usepackage{latexsym}
\usepackage{xspace}

\usepackage{graphicx}

\usepackage{xspace}

\input diagxy
\xyoption{color}

\usepackage{algorithm}

\usepackage[noend]{algpseudocode}

\newtheorem{theorem}{Theorem}[section]
\newtheorem{lemma}[theorem]{Lemma}
\newtheorem{corollary}[theorem]{Corollary}

\newtheorem{definition}{Definition}[section]
\newtheorem{example}{Example}[section]

\newtheorem{remark}{Remark}[section]
\newtheorem*{acknnowledgments}{Acknowledgments}

\usepackage{tikz}
\usetikzlibrary{automata,positioning}

\usetikzlibrary{calc}
\usetikzlibrary{arrows, decorations.markings}
\usetikzlibrary{patterns}
\tikzstyle{vecArrow} = [thick, decoration={markings,mark=at position
   1 with {\arrow[semithick]{open triangle 60}}},
   double distance=1.4pt, shorten >= 5.5pt,
   preaction = {decorate},
   postaction = {draw,line width=1.4pt, white,shorten >= 4.5pt}]
   
\tikzstyle{vecArrow2} = [thick, black, double distance=1.4pt, shorten >= 1.5pt,
preaction = {decorate},
   postaction = {draw,line width=1.4pt, white,shorten >= 1.5pt}]
   
\tikzstyle{innerWhite} = [semithick, white,line width=1.4pt, shorten >= 4.5pt]

\makeatletter
\def\slashedarrowfill@#1#2#3#4#5{%
  $\m@th\thickmuskip0mu\medmuskip\thickmuskip\thinmuskip\thickmuskip
  \relax#5#1\mkern-7mu%
  \cleaders\hbox{$#5\mkern-2mu#2\mkern-2mu$}\hfill
  \mathclap{#3}\mathclap{#2}%
  \cleaders\hbox{$#5\mkern-2mu#2\mkern-2mu$}\hfill
  \mkern-7mu#4$%
}
\def\rightslashedarrowfill@{%
  \slashedarrowfill@\relbar\relbar\mapstochar\rightarrow}
\newcommand\xslashedrightarrow[2][]{%
  \ext@arrow 0055{\rightslashedarrowfill@}{#1}{#2}}
\makeatother

\def\slashedrightarrow{\xslashedrightarrow{}}

\newcommand*\circled[1]{\tikz[baseline=(char.base)]{
            \node[shape=circle,draw,inner sep=2pt] (char) {#1};}}

\newcommand{\group}[1]  {
  \mathbb{#1}
}

\newcommand{\catl}[1]  {
  \mathbb{#1}
}
\newcommand{\catw}[1]  {
  \mathbf{#1}
}
\newcommand{\word}[1]  {
  \mathit{#1}
}
\newcommand{\mor}[3]  {
  #1 \colon #2 \rightarrow #3
}
\newcommand{\dist}[3]  {
  #1 \colon #2 \slashedrightarrow #3
}
\newcommand{\tuple}[1]  {
	\langle #1 \rangle
}
\newcommand{\id}[1]  {
  \word{id}_{#1}
}
\newcommand{\cont}[1]  {
  \catw{Cont}(\group{#1})
}
\newcommand{\classifying}[1]  {
  \catw{Set}[T]
}
\newcommand{\aut}[1]  {
  \mathit{Aut}(#1)
}

\newcommand{\struct}[1]  {
  \mathcal{#1}
}

\author{Micha{\l} R. Przyby{\l}ek}

\title{Beyond sets with atoms:\\Definability in first order logic}
\date{\today}

\begin{document}
\sloppy

\maketitle

\abstract{Sets with atoms serve as an alternative to ZFC foundations for mathematics, where some infinite, though highly symmetric sets, behave in a finitistic way. Therefore, one can try to carry over analysis of the classical algorithms from finite structures to symmetric infinite structures. Recent results show that this is indeed possible and leads to many practical applications: automata over infinite alphabets \cite{DBLP:conf/popl/BojanczykBKL12}, model checking \cite{DBLP:conf/csl/KlinL17}, constraint satisfaction solving \cite{ochremiak2016extended}, \cite{DBLP:conf/lics/KlinKOT15}, programming languages \cite{klin2016smt}, \cite{Kopczynski:2017:LSS:3093333.3009876}  and \cite{cheney2008nominal}, to name a few. In this paper we shall take another route to finite analysis of infinite sets, which extends and sheds more light on sets with atoms. As an application of our theory we give a characterisation of languages recognized by automata definable in fragments of first order logic.}

\section{Introduction}
\label{s:introduction}
In the late '70s Stephen Schanuel working on the theory of combinatorial functions studied the topos of pullback preserving functors from the category of finite sets and injections to the category of sets, which is nowadays known as the Shanuel topos. Shortly afterwards, when the theory of classifying toposes emerged, it has been discovered that the Schanuel topos is the classifying topos for the first order theory of infinite decidable objects\footnote{Classically, this is just the theory of pure sets. See Example~\ref{e:infinite:decidable}.} \cite{wraith}. 

The Schanuel topos was then rediscovered by James Gabbay and Andrew Pitts \cite{gabbay1999new} as an elegant formalism for reasoning about name bindings in formal languages. This idea was further pursued \cite{pitts2013nominal} and the Schanuel topos earned a new name --- the topos of \emph{nominal sets} --- starting a completely new life in theoretical computer science. A decade later, the connection between nominal sets and the theory of classifying toposes was forgoten and some of the classical results were discovered again in \cite{DBLP:conf/lics/BojanczykKL11}, \cite{DBLP:conf/lics/BojanczykKLT13} and again in \cite{bojanbook}. Nonetheless, many well-known classical results are still unknown.

\begin{figure}[tb]
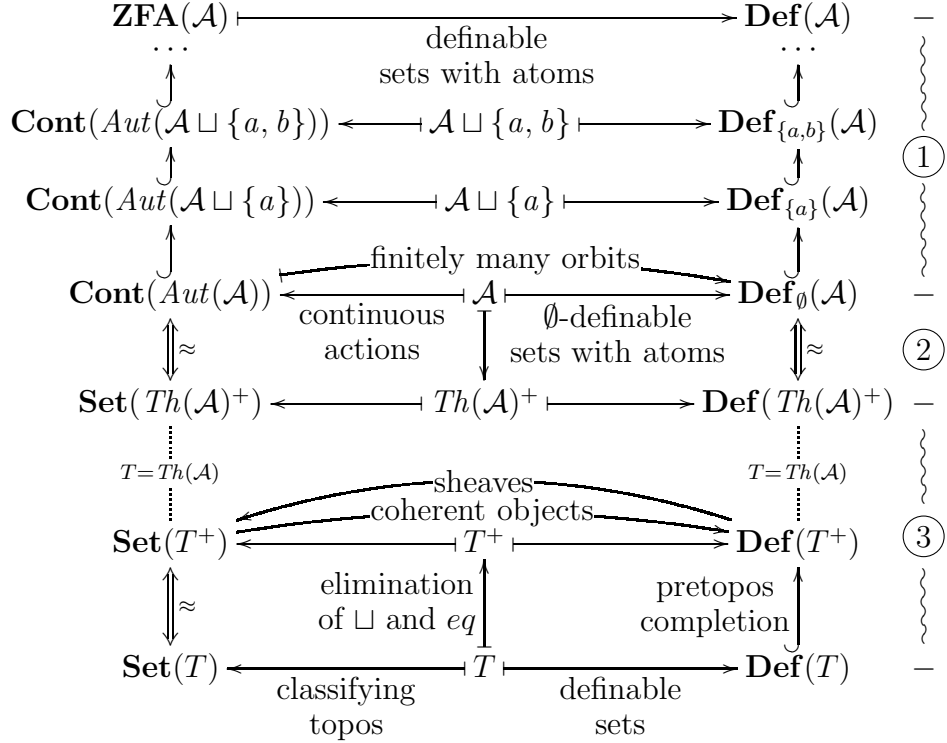

$$\bfig
\node c_topos(-200, -200)[\catw{Set}(T)]
\node c_theory(800, -200)[T]
\node c_def(1800, -200)[\catw{Def}(T)]

\node c_topos_plus(-200, 200)[\catw{Set}(T^+)]
\node c_theory_plus(800, 200)[T^+]
\node c_def_plus(1800, 200)[\catw{Def}(T^+)]

\arrow|l|/<-|/[c_topos`c_theory;\txt{classifying\\topos}]
\arrow|l|/|->/[c_theory`c_def;\txt{definable\\sets}]

\arrow/|->/[c_theory_plus`c_topos_plus;]
\arrow/|->/[c_theory_plus`c_def_plus;]

\arrow|l|/|->/[c_theory`c_theory_plus;\txt{elimination\\of $\sqcup$ and $eq$}]
\arrow|l|/{^{(}->}/[c_def`c_def_plus;\txt{pretopos\\completion}]
\arrow|r|/<=>/[c_topos`c_topos_plus;\approx]
\arrow|m|/@/^1pc//[c_topos_plus`c_def_plus;\txt{coherent objects}]
\arrow|m|/@/_2pc//[c_def_plus`c_topos_plus;\txt{sheaves}]

\node c_topos_th(-200, 650)[\catw{Set}(\word{Th}(\struct{A})^+)]
\node c_theory_th(800, 650)[\word{Th}(\struct{A})^+]
\node c_def_th(1800, 650)[\catw{Def}(\word{Th}(\struct{A})^+)]

\arrow/|->/[c_theory_th`c_topos_th;]
\arrow/|->/[c_theory_th`c_def_th;]

\node cont_topos(-200, 1000)[\textcolor{black}{\cont{\aut{\struct{A}}}}]
\node struct(800, 1000)[\textcolor{black}{\struct{A}}]
\node def(1800, 1000)[\textcolor{black}{\catw{Def}_\emptyset(\struct{A})}]

\arrow/@{|->}@[black]/[struct`c_theory_th;]
\arrow|r|/@{<=>}@[black]/[cont_topos`c_topos_th;\textcolor{black}{\approx}]
\arrow|r|/@{<=>}@[black]/[def`c_def_th;\textcolor{black}{\approx}]

\arrow|l|/@{<-|}@[black]/[cont_topos`struct;\textcolor{black}{\txt{continuous\\actions}}]
\arrow|l|/@{|->}@[black]/[struct`def;\textcolor{black}{\txt{$\emptyset$-definable\\sets with atoms}}]
\arrow|m|/@/^1.1pc/@{|->}@[black]/[cont_topos`def;\textcolor{black}{\txt{finitely many orbits}}]


\arrow|m|/./[c_topos_th`c_topos_plus;T = \word{Th}(\struct{A})]
\arrow|m|/./[c_def_th`c_def_plus;T = \word{Th}(\struct{A})]

\node col1(-200, 1300)[\textcolor{black}{\cont{\aut{\struct{A} \sqcup \word{\{a\}}}}}]
\node col3(-200, 1550)[\textcolor{black}{\cont{\aut{\struct{A} \sqcup \word{\{a, b\}}}}}]
\node col4(-200, 1790)[\textcolor{black}{\cdots}]
\node zfa(-200, 1890)[\textcolor{black}{\catw{ZFA}(\struct{A})}]

\node dcol1(1800, 1300)[\textcolor{black}{\catw{Def}_{\{a\}}(\struct{A})}]
\node dcol3(1800, 1550)[\textcolor{black}{\catw{Def}_{\{a, b\}}(\struct{A})}]
\node dcol4(1800, 1790)[\textcolor{black}{\cdots}]
\node dzfa(1800, 1890)[\textcolor{black}{\catw{Def}(\struct{A})}]

\node struct1(850, 1300)[\textcolor{black}{\struct{A \sqcup \word{\{a\}}}}]
\node struct2(850, 1550)[\textcolor{black}{\struct{A \sqcup \word{\{a, b\}}}}]

\arrow|l|/@{<-|}@[black]/[col1`struct1;]
\arrow|l|/@{<-|}@[black]/[dcol1`struct1;]

\arrow|l|/@{<-|}@[black]/[col3`struct2;]
\arrow|l|/@{<-|}@[black]/[dcol3`struct2;]

\arrow/@{^{(}->}@[black]/[cont_topos`col1;]
\arrow/@{^{(}->}@[black]/[col1`col3;]
\arrow/@{^{(}->}@[black]/[col3`col4;]

\arrow/@{^{(}->}@[black]/[def`dcol1;]
\arrow/@{^{(}->}@[black]/[dcol1`dcol3;]
\arrow/@{^{(}->}@[black]/[dcol3`dcol4;]
\arrow|r|/@{|->}@[black]/[zfa`dzfa;\textcolor{black}{\txt{definable\\sets with atoms}}]



\node one(2200, 1890)[-]
\node two(2200, 1000)[-]
\node three(2200, 650)[-]
\node four(2200, -200)[-]

\arrow|m|/~/[one`two;\txt{\circled{1}}]
\arrow|m|/~/[two`three;\txt{\circled{2}}]
\arrow|m|/~/[three`four;\txt{\circled{3}}]

\efig$$
\caption{Correspondence between classifying toposes, nominal sets and set theories with atoms.}
\label{f:topos:theory}
\end{figure}

This paper presents nominal sets, and their older cousins: sets with atoms, as a part of a bigger picture (see Figure~\ref{f:topos:theory}, which will be explained throughout the paper) --- the theory of classifying toposes for the positive existential fragment of intuitionistic first order logic. According to this picture, generalised nominal sets are precisely the classifying toposes for $\omega$-categorical structures, whereas set theories with atoms are precisely the filtered colimits of some canonical diagrams of generalised nominal sets.
We shall focus on the aspects of computability in positive existential logic --- which algorithms can be effectively executed, when the domains of the variables are interpreted as ``potentially infinite'' definable sets. This goes beyond theories of oligomorphic structures (Example~\ref{e:pure:sets} and Example~\ref{e:rationals}). Our framework is suitable for $\omega$-categorical structures, which are not oligomorphic (Example~\ref{e:multi:omegacat}), structures build from $\omega$-categorical structures by adding infinitely many constants (Example~\ref{e:pure:sets:constants}), classical non-complete theories (Example~\ref{e:dense:order}), intuitionistic propositional theories (Example~\ref{e:prop:theory} and Example~\ref{e:impossible:theory}), and many more.

\begin{example}[Pure sets]\label{e:pure:sets}
Let $\struct{N} = \{0,1,2,\dotsc\}$ be a countably infinite set over empty signature $\Xi$. Then the first order theory of $\struct{N}$ is $\omega$-categorical, i.e.~there is exactly one countable model of the theory up to an isomorphism.  This theory is called the theory of ``pure sets''.
\end{example}

\begin{example}[Rational numbers with ordering]\label{e:rationals}
Let $\struct{Q} = \tuple{Q, {\leq}}$ be the structure whose universe is interpreted as the set of rational numbers $Q$ with a single binary relation ${\leq} \subseteq Q \times Q$ interpreted as the natural ordering of rational numbers. Then the first order theory of $\struct{Q}$ is $\omega$-categorical.
\end{example}

\begin{example}[Multi-sorted $\omega$-categorical theory]\label{e:multi:omegacat}
Let $\struct{S}$ be a structure consisting of countably many countable sorts identified with natural numbers $N$ and such that the $i$-th sort interprets constants $\{0, 1, \cdots, i-1\}$. Then the theory  $\word{Th}(\struct{S})$ is $\omega$-categorical. However, the group of automorphisms $\aut{\struct{S}}$ of $\struct{S}$ in not oligomorphic ---  since the automorphisms act independently on each sort, the group has infinitely many orbits. 
\end{example}

\begin{example}[Pure sets with constants]\label{e:pure:sets:constants}
Let $\struct{N} \sqcup N$ be the structure from Example~\ref{e:pure:sets} over an extended signature consisting of all constants $n \in N$. Then the first order theory of $\struct{N} \sqcup N$  has countably many non-isomorphic countable models.
\end{example}

\begin{example}[Dense linear order]\label{e:dense:order}
Let $T$ be the first order theory of dense linear orders, i.e.: it is a theory, over signature consisting of a single binary predicate ${<}$, with the following axioms (written as first order sequents):
\begin{eqnarray*}
a < a & \vdash & \bot \\
a < b \land b < c & \vdash & a < c \\
& \vdash & a < b \lor b < a \lor a = b \\
a < c & \vdash & \exists_b \; a < b \land b < c
\end{eqnarray*}
This theory is not complete, as it does not specify whether a given linear order has the smallest and the largest element, and if so, whether or not they coincide. 
\end{example}

\begin{example}[Propositional theory with one variable]\label{e:prop:theory}
By propositional theory with one variable we shall mean the empty positive existential theory over zero-sorted signature $\Xi_1$  with a single nullary relation $p \subseteq |\Xi_1|^0 = 1$. A model of this theory in any topos is an internal truth value (i.e.~subobject of the terminal object). For example, in $\catw{Set}$ there are exactly two models: one in which $p$ is false, and another in which $p$ is true.
\end{example}

\begin{example}[Seemingly impossible theory]\label{e:impossible:theory}
By seemingly impossible theory $I$ we shall mean the positive existential theory over zero-sorted signature $\Xi_N$ with countably many nullary relations $\{-n\}_{n \in \mathcal{N}}$ with following axioms: ${-(n+1) \vdash -n}$. A model of $I$ in $\catw{Set}$ consists of a  sequence of true values, followed by zero or infinitely many false values, i.e.~a model can be described by an element of $1^\infty \cup \{1^*0^\infty\}$ and identified with an extended natural number $k$. Under such identification, a (necessarily unique) homomorphism $k \rightarrow l$ exists iff $k \leq l$. Therefore, the category of $I$-models in $\catw{Set}$ is (equivalent to) the poset of extended natural numbers with their usual ordering. 
\end{example}

In \cite{DBLP:journals/corr/BojanczykKL14} a concept of a while-program with semantics in definable sets with atoms $\struct{A}$ has been defined. The authors examine conditions on $\struct{A}$ that ensure that certain while-programs terminate. As an illustrative example, consider the reachability problem on directed graphs. A while-program for this problem is presented as Algorithm~\ref{a:reachability}. This algorithm can be actually implemented in a natural way in a programming language that supports computaton on sets with atoms, for instance: LOIS or $N\lambda$\footnote{A working implementation of $N\lambda$, a functional programming language capable of processing infinite structures with atoms, is available through the web-site: \url{https://www.mimuw.edu.pl/~szynwelski/nlambda/}.} (see \cite{lois} and \cite{lambda}, also \cite{klin2016smt}, \cite{Kopczynski:2017:LSS:3093333.3009876}  and \cite{cheney2008nominal} for more details). We will see in Section~\ref{ss:zfa} that by transfer principle (Theorem~\ref{t:transfer}), the program can be actually executed in the category $\cont{\aut{\struct{A} \sqcup \word{A_0}}}$ of continuous actions of the topological group of automorphisms of structure $\struct{A} \sqcup \word{A_0}$ for some finite $A_0 \subset A$. Moreover, the conditions the authors examine imply that $\struct{A}$ is oligomorphic and  $\cont{\aut{\struct{A} \sqcup \word{A_0}}}$ is the classifying topos for the theory of $\struct{A} \sqcup \word{A_0}$. Therefore, (see Section~\ref{ss:classifying:topos}) their framework restricts to sets definable in the first order theory of oligomorphic structure $\struct{A} \sqcup \word{A_0}$. We will see that Algorithm~\ref{a:reachability} can be effectively executed on sets definable in theories from all our Examples~\ref{e:pure:sets}, \ref{e:rationals}, \ref{e:multi:omegacat}, \ref{e:pure:sets:constants}, \ref{e:dense:order},  \ref{e:prop:theory} and \ref{e:impossible:theory}. 
\begin{algorithm}[tb]
\begin{algorithmic}
\Procedure{REACHABLE}{$E, a, b$}
\State $R' \gets \emptyset$
\State $R \gets \{b\}$
\While{$R' \neq R$}
	\If{$a \in R$}
		\Return $\top$
	\EndIf
	\State $R' \gets R$
	\For{$\tuple{x, y} \in E$}
	\If{$y \in R'$}
	 \State $R \gets R \cup \{x\}$
	\EndIf
	\EndFor
\EndWhile
\Return $\bot$
\EndProcedure
\end{algorithmic}    
\caption{Reachability algorithm} 
\label{a:reachability} 
\end{algorithm}

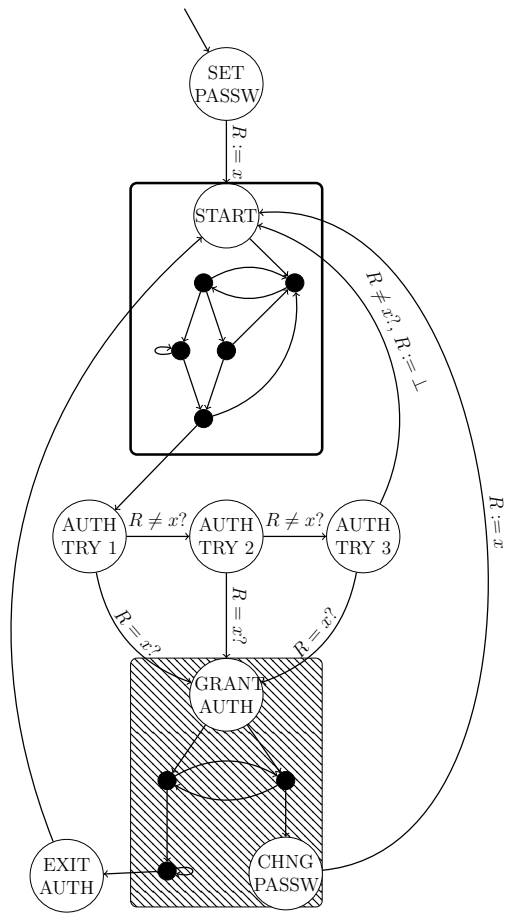
\begin{figure}[htb]
    \centering
    \resizebox{0.6\textwidth}{!}{%

    \begin{tikzpicture}[sloped,ball/.style = {circle, draw, align=center, anchor=north, inner sep=0}]
    
    \draw[black,ultra thick,rounded corners] (0.9,-5.5) rectangle (5.1,-11.5);
    \draw[rounded corners,pattern=north west lines] (0.9,-16) rectangle (5.1,-21.5);
\node[text width=1.4cm] (init) at (2,-1.5) {{}};
\node[ball,text width=1.4cm] (passw) at (3,-2.5) {{SET\\PASSW}};
\node[ball,text width=1.4cm] (start) at (3,-5.5) {{START}};

\draw[thick,->]  (init) to  (passw) ;
\draw[thick,->]  (passw) to  node[midway,above] {$R := x$} (start);

\node[ball,text width=1.4cm] (try1) at (0,-12.5) {{AUTH\\TRY 1}};
\node[ball,text width=1.4cm] (try2) at (3,-12.5) {{AUTH\\TRY 2}};
\node[ball,text width=1.4cm] (try3) at (6,-12.5) {{AUTH\\TRY 3}};

\node[ball,text width=1.4cm,fill=white] (access) at (3,-16) {{GRANT\\AUTH}};

\draw[thick,->]  (try1) to [bend right] node[midway,sloped,above] {$R = x$?} (access);
\draw[thick,->]  (try1) to  node[midway,sloped,above] {$R \neq x$?} (try2);
\draw[thick,->]  (try2) to  node[midway,sloped,above] {$R \neq x$?} (try3);
\draw[thick,->]  (try2) to  node[midway,sloped,above] {$R = x$?} (access);
\draw[thick,<-]  (access) to [bend right] node[midway,above] {$R = x$?} (try3);

\draw[thick,->]  (try3) to [bend right=50] node[midway,sloped,above] {$R \neq x?$, $R := \bot$} (start);

\node[ball,text width=1.4cm] (exit) at (-0.5,-20) {{EXIT\\AUTH}};
\draw[thick,->]  (exit) to [bend left=35] (start);

\node[ball,text width=0.4cm,fill=black] (a1) at (2.5,-7.5) {{}};
\node[ball,text width=0.4cm,fill=black] (a2) at (4.5,-7.5) {{}};
\node[ball,text width=0.4cm,fill=black] (a3) at (2,-9) {{}};
\node[ball,text width=0.4cm,fill=black] (a4) at (3,-9) {{}};
\node[ball,text width=0.4cm,fill=black] (a5) at (2.5,-10.5) {{}};

\draw[thick,->]  (a4) to (a5);

\draw[thick,->]  (start) to (a2);
\draw[thick,->]  (a1) to (a3);
\draw[thick,->]  (a1) to (a4);
\draw[thick,->]  (a3) to (a5);
\draw[thick,->]  (a4) to (a2);
\draw[thick,->]  (a5) [bend right=40] to (a2);
\draw[thick,->]  (a2) [bend left] to (a1);
\draw[thick,->]  (a1) [bend left] to (a2);


\draw[thick,->]  (a3) to [loop left] (a3);
\draw[thick,->]  (a5) to (try1);

\node[ball,text width=0.4cm,fill=black] (b1) at (1.7,-18.5) {{}};
\node[ball,text width=0.4cm,fill=black] (b2) at (4.3,-18.5) {{}};
\node[ball,text width=0.4cm,fill=black] (b3) at (1.7,-20.5) {{}};

\node[ball,text width=1.4cm,fill=white] (change) at (4.3,-19.95) {{CHNG\\PASSW}};
\draw[thick,->]  (change) to [bend right=90] node[midway,sloped,above] {$R := x$} (start);

\draw[thick,->]  (access) to (b1);
\draw[thick,->]  (access) to (b2);
\draw[thick,->]  (b1) to (b3);
\draw[thick,->]  (b2) to (change);
\draw[thick,->]  (b3) to (exit);
\draw[thick,->]  (b2) [bend left] to (b1);
\draw[thick,->]  (b1) [bend left] to (b2);


\draw[thick,->]  (b3) to [loop right] (b3);

\end{tikzpicture} 
    }%
    \caption{A register machine that models access control to the dashed part of the system.}
    \label{f:register:machine}
  \end{figure}

In parallel with Algorithm~\ref{a:reachability}, we shall study the languages that can be recognised by a generalisation of finite memory machines in the sense of Kaminski and Francez \cite{kaminski1994finite}. An example of such a  machine is presented on Figure~\ref{f:register:machine}. The machine has a single register $R$ and can test  for equality and inequality only. It starts in state ``SET PASSW'', where it awaits for the user to provide a password $x$. This password is then stored in register $R$, and the machine enters state ``START''. Inside the top rectangle the machine can perform actions that do not require authentication, whereas the actions that require authentication are presented inside the bottom rectangle. The bottom rectangle can be entered by state ``GRANT AUTH'', which can be accessed from one of three authentication states. In order to authorise, the machine moves to state ``AUTH TRY 1'', where it gets input $x$ from the user. If the input is the same as the value previously stored in register $R$, then the machine enters state ``GRANT AUTH''. Otherwise, it moves to state ``AUTH TRY 2'' and repeats the procedure. Upon second unsuccessful authorisation, the machine moves to state ``AUTH TRY 3''. But  if the user provides a wrong password when the machine is in state ``AUTH TRY 3'', the register $R$ is erased (replaced with a value that is outside of the user's alphabet) --- preventing the machine to reach any of the states from the bottom rectangle. Inside the bottom rectangle any action that requires authentication can be performed. For example, the user may request the change of the password.

The authors of \cite{DBLP:conf/lics/BojanczykKLT13} found that finite memory machines correspond to automata definable in the theory of pure sets from Example~\ref{e:pure:sets}. A suitable generalisation of this concept to positive existential theories is presented in Section~\ref{s:automata}. We prove there a version of Myhill-Nerode theorem for languages of subcompact/definable deterministic automata (Theorem~\ref{t:subcompact}, Corollary~\ref{c:compact:mn} and Theorem~\ref{t:definable:mn}), and a characterisation theorem for definable non-deterministic automata by definable relational monoids (Corollary~\ref{c:pro:regular}). 
Notice that one has to be careful when studying languages recognised by automata definable in a logical theory, because a language recognized by a definable automata is almost never definable. This problem can be overtaken by describing a language $L$ as a collection of languages $(L^{*k})_{k \in N}$, where $L^{*k}$ is definable and consists of these words of $L$ whose length is at most $k$. Nonetheless, the full justification of such a definition is difficult without the theory of classifying toposes, and the explicit calculations are messy. Therefore, we review some basic facts about classifying toposes in Subsection~\ref{ss:classifying:topos} and perform all of the necessary computations in Section~\ref{s:automata} inside the classifying topos of the theory, where the concept of a language can live naturally. For more information about automata in categories we refer to \cite{adamek1990automata}, \cite{adamek1974free}, \cite{eilenberg1967automata} and \cite{eilenberg1974automata}.

Section~\ref{s:roadmap} is devoted to explaining Figure~\ref{f:topos:theory}. Subsection~\ref{ss:zfa} explains the left side of part \circled{1} on the picture: how Zermelo-Frankel set theory with atoms can be constructed from toposes of continuous actions of topological groups. We state here a meta-theorem (Theorem~\ref{t:transfer} together with Theorem~\ref{t:extended:transfer}) allowing us to delegate computations from ZFA to toposes of continuous actions of topological groups. The right side of \circled{1} on the picture together with \circled{2} is explained in Subsection~\ref{ss:first:order:structure}. We investigate there possible extensions to definability in sets with atoms and prove Theorem~\ref{t:stage} indicating why such attempts might be futile in general. The right square of \circled{3} is explained in Subsection~\ref{ss:coherent:theory}, where we study definability in positive existential theories. Finally, the outer square of \circled{3} is roughly explained in Subsection~\ref{ss:classifying:topos}; for more information about Grothendieck toposes we refer the reader to \cite{maclane2012sheaves}, \cite{johnstone2003sketches} and \cite{borceux}.


\section{The roadmap}
\label{s:roadmap}
In this section we shall explain Figure~\ref{f:topos:theory} in detail and discuss how the computations with atoms can be carried over to a more general framework of classifying toposes for positive existential theories.
\subsection{Set with atoms}\label{ss:zfa}
Let $\struct{A}$ be an algebraic structure (both operations and relations are allowed) with universum $A$. We shall think of elements of $\struct{A}$ as ``atoms''. A von Neumann-like hierarchy $V_\alpha(\struct{A})$ of sets with atoms $\struct{A}$ can be defined by transfinite induction \cite{mostowski1939unabhangigkeit}, \cite{halbeisen2017combinatorial}:
\begin{itemize}
    \item $V_0(\struct{A}) = A$
    \item $V_{\alpha + 1}(\struct{A}) = \mathcal{P}(V_{\alpha}(\struct{A})) \cup V_{\alpha}(\struct{A})$
    \item $V_{\lambda}(\struct{A}) = \bigcup_{\alpha < \lambda} V_{\alpha}(\struct{A})$ if $\lambda$ is a limit ordinal
\end{itemize}
Then the cumulative hierarchy of sets with atoms $\struct{A}$ is just $V(\struct{A}) = \bigcup_{\alpha \colon \mathit{Ord}} V_{\alpha}(\struct{A})$. Observe, that the universe $V(\struct{A})$ carries a natural action $\mor{(\bullet)}{\aut{\struct{A}} \times V(\struct{A})}{V(\struct{A})}$ of the automorphism group $\aut{\struct{A}}$ of structure $\struct{A}$ --- it is just applied pointwise to the atoms of a set. If $X \in V(\struct{A})$ is a set with atoms then by its set-wise stabiliser we shall mean the set: $\aut{\struct{A}}_X = \{\pi \in \aut{\struct{A}} \colon \pi \bullet X = X \}$; and by its point-wise stabiliser the set: $\aut{\struct{A}}_{(X)} = \{\pi \in \aut{\struct{A}} \colon \forall_{x \in X} \pi \bullet x = x \}$. Moreover, for every $X$, these sets inherit a group structure from $\aut{\struct{A}}$.

There is an important sub-hierarchy of the cumulative hierarchy of sets with atoms $\struct{A}$, which consists of ``symmetric sets'' only. To define this hierarchy, we have to equip $\aut{\struct{A}}$ with the structure of a topological group. A set $X \in V(\struct{A})$ is \emph{symmetric} if the set-wise stabilisers of all of its descendants $Y$ is an open set (an open subgroup of $\aut{\struct{A}}$), i.e.~for every $Y \in^* X$ we have that: $\aut{\struct{A}}_{Y}$ is open in $\aut{\struct{A}}$, where ${\in^*}$ is the reflexive-transitive closure of the membership relation ${\in}$. A function between symmetric sets is called symmetric if its graph is a symmetric set.
Of a special interest is the topology on $\aut{\struct{A}}$ inherited from the product topology on $\prod_A A = A^A$ (i.e.~the Tychonoff topology). We shall call this topology the canonical topology on $\aut{\struct{A}}$. In this topology, a subgroup $\group{H}$ of $\aut{\struct{A}}$ is open if there is a finite $A_0 \subseteq A$ such that: $\aut{\struct{A}}_{(A_0)} \subseteq \group{H}$, i.e.: group $\group{H}$ contains a pointwise stabiliser of some finite set of atoms. The sub-hierarchy of $V(\struct{A})$ that consists of symmetric sets according to the canonical topology on $\aut{\struct{A}}$ will be denoted by $\catw{ZFA}(\struct{A})$ (it is a model of Zermelo-Fraenkel set theory with atoms).

\begin{remark}
The above definition of hierarchy of symmetric sets is equivalent to another one used in model theory. 
By a normal filter of subgroups of a group $\group{G}$ we shall understand a filter $\mathcal{F}$ on the poset of subgroups of $\group{G}$ closet under conjugation, i.e.~if $g \in \group{G}$ and $\group{H} \in \mathcal{F}$ then $g \group{H} g^{-1} = \{g \bullet h \bullet g^{-1} \colon h \in \group{H}\} \in \mathcal{F}$. Let $\mathcal{F}$ be a normal filter of subgroups of $\aut{\struct{A}}$. We say that a set $X \in V(\struct{A})$ is $\mathcal{F}$-symmetric if the set-wise stabilisers of all of its descendants $Y$ belong to $\mathcal{F}$ --- i.e. $Y \in^* \mathcal{F}$. To see that the definitions of symmetric sets and $\mathcal{F}$-symmetric sets are equivalent, observe first that if $\group{G}$ is a topological group, then the set $\mathcal{F}$ of all open subgroups of $\group{G}$ is a normal filter of subgroups. In the other direction, if $\mathcal{F}$ is a normal filter of subgroups of a group $\group{G}$, then we may define a topology on $\group{G}$ by declaring sets $U \subseteq \group{G}$ to be open if they satisfy the following property: for every $g \in U$ there exists $\group{H} \in \mathcal{F}$ such that $g \group{H} \subseteq U$. According to this topology a group $\group{U}$ is open iff  $\group{U} \in \mathcal{F}$ --- just observe that for every group $\group{U}$ and for every $g \in \group{U}$ we have that $g\group{U} = \group{U}$; and if $\group{H} \in \mathcal{F}$ such that $\group{H} = 1\group{H} \subseteq \group{U}$ then by the property of the filter, $\group{U} \in \mathcal{F}$.    
\end{remark}

\begin{example}[The basic Fraenkel-Mostowski model]\label{e:first:zfa}
Let $\struct{N}$ be the structure from Example~\ref{e:pure:sets}. We call $\catw{ZFA}(\struct{N})$ the basic Fraenkel-Mostowski model of set theory with atoms. Observe that $\aut{\struct{N}}$ is the group of all bijections (permutations) on $N$. 
The following are examples of sets in $\catw{ZFA}(\struct{N})$:
\begin{itemize}
\item all sets without atoms, e.g.~$\emptyset, \{\emptyset\}, \{\emptyset, \{\emptyset\}, \dotsc\}, \dotsc$
\item all finite subsets of $N$, e.g.~$\{0\}, \{0,1,2,3\}, \dotsc$
\item all cofinite subsets of $N$, e.g.~$\{1, 2, 3, \dotsc\}, \{4, 5, 6, \dotsc\}, \dotsc$
\item $N\times N$
\item $\{\tuple{a,b} \in N^2 \colon a \neq b\}$
\item $N^* = \bigcup_{k\in N} N^k$ 
\item $\mathcal{K}(N) = \{N_0 \colon N_0 \subseteq N, \textit{$N_0$ is finite}\}$
\item $\mathcal{P}_s(N) = \{N_0 \colon N_0 \subseteq N, \textit{$N_0$ is symmetric}\}$
\end{itemize}
Here are examples of sets in $V(\struct{N})$ which are not symmetric:
\begin{itemize}
\item $\{0, 2, 4, 6, \dotsc\}$
\item $\{\tuple{n, m} \in N^2 \colon n \leq m\}$
\item the set of all functions from $N$ to $N$
\item $\mathcal{P}(N) = \{N_0 \colon N_0 \subseteq N\}$
\end{itemize}
\end{example}

\begin{example}[The ordered Fraenkel-Mostowski model]\label{e:ordered:zfa}
Let $\struct{Q}$ be the structure from Example~\ref{e:rationals}. We call $\catw{ZFA}(\struct{Q})$ the ordered Fraenkel-Mostowski model of set theory with atoms. Observe that $\aut{\struct{Q}}$ is the group of all order-preserving bijections on $Q$. 
All symmetric sets from Example~\ref{e:first:zfa} are symmetric sets in $\catw{ZFA}(\struct{Q})$ when $N$ is replaced by $Q$. Here are some further symmetric sets:
\begin{itemize}
\item $\{\tuple{p, q} \in Q^2 \colon p \leq q\}$
\item $\{\tuple{p, q} \in Q^2 \colon 0 \leq p \leq q \leq 1 \}$
\end{itemize}
\end{example}

\begin{example}[The second Fraenkel-Mostowski model]\label{e:second:zfa}
Let $\struct{S} = \tuple{Z^*, {-}, (|{-}|_n)_{n \in N}}$ be the structure of non-zero integer numbers, with unary ``minus'' operation $\mor{(-)}{Z^*}{Z^*}$ and with unary relations ${|{-}|_n} \subseteq Z^*$ defined in the following way: $|z|_n \Leftrightarrow |z| = n$. We call $\catw{ZFA}(Z^*)$ the second Fraenkel-Mostowski model of set theory with atoms. Observe that $\aut{\struct{Z^*}} \approx \group{Z}_2^N$, 
therefore the following sets are symmetric in $\catw{ZFA}(Z^*)$:
\begin{itemize}
\item $\{\dotsc, -6, -4, -2, 2, 4, 6, \dotsc\}$
\item $\{\tuple{x, y} \in Z^* \times Z^* \colon x = 3y \}$
\end{itemize}
\end{example}

Observe that the group $\aut{\struct{A}}_{(A_0)}$ is actually the group of automorphism of structure $\struct{A}$ extended with constants $A_0$, i.e.: $\aut{\struct{A}}_{(A_0)} = \aut{\struct{A} \sqcup A_0}$. Then a set $X \in V(\struct{A})$ is symmetric if and only if there is a finite $A_0 \in A$ such that $\aut{\struct{A} \sqcup A_0} \subseteq \aut{\struct{A}}_X$ and the canonical action of topological group $\aut{\struct{A} \sqcup A_0}$ on discrete set $X$ is continuous. A symmetric set is called $A_0$-equivariant (or equivariant in case $A_0 = \emptyset$) if $\aut{\struct{A} \sqcup A_0} \subseteq \aut{\struct{A}}_X$. Therefore, the (non-full) subcategory of $\catw{ZFA}(\struct{A})$ on $A_0$-equivariant sets and $A_0$-equivariant functions (i.e.~functions whose graphs are $A_0$-equivariant) is equivalent to the category $\cont{\aut{\struct{A} \sqcup \word{A_0}}} \subseteq \catw{Set}^{\aut{\struct{A} \sqcup \word{A_0}}}$ of continuous actions of the topological group $\aut{\struct{A} \sqcup A_0}$ on discrete sets.

\begin{example}[Equivariant sets]
In the basic Fraenkel-Mostowski model:
\begin{itemize}
\item all sets without atoms are equivariant
\item all finite subsets $N_0 \subseteq N$ are $N_0$-equivariant
\item all finite subsets $N_0 \subseteq N$ are $(N \setminus N_0)$-equivariant
\item $N\times N, N^{(2)}, \mathcal{K}(N), \mathcal{P}_S(N)$ are equivariant
\end{itemize}
\end{example}

In many works on computations in sets with atoms (e.g.~Remark~20 in \cite{DBLP:conf/lics/KlinKOT15}, Theorem~11.1 in \cite{lmcs:5877}, \dots) the authors focus on equivariant sets and equivariant functions (i.e.~the category $\cont{\aut{\struct{A}}}$) and claim that their results carry over to $\catw{ZFA}(\struct{A})$. We shall now give a formal argument why such claims are valid.
\begin{lemma}[Presentation of $\catw{ZFA}(\struct{A})$]\label{l:presentation}
Let $\struct{A}$ be an algebraic structure. Then there is a functor $\mor{\Theta}{K(A)}{\word{Log}}$ from the poset $K(A)$ of finite subsets of $A$ seen as a posetal category to the 2-category $\catw{Log}$ of elementary toposes and logical functors. This functor maps $A_0$ to $\cont{\aut{\struct{A} \sqcup \word{A_0}}}$ and ${A_0 \subseteq A_1}$ to the logical embedding: $\cont{\aut{\struct{A} \sqcup \word{A_0}}} \rightarrow  \cont{\aut{\struct{A} \sqcup \word{A_1}}}$. Moreover, $\catw{ZFA}(\struct{A})$ is the colimit of $\Theta$ in $\catw{Log}$ ---  the canonical embeddings $\cont{\aut{\struct{A} \sqcup \word{A_0}}} \rightarrow \catw{ZFA}(\struct{A})$ for finite $A_0 \subseteq A$ are logical embeddings (i.e.~preserve elementary topos structure).
\end{lemma}

Because the forgetful functor from $\catw{Log}$ to the category of locally small categories $\catw{Cat}$ preserves filtered colimits, $\catw{ZFA}(\struct{A})$ is also the filtered colimit of logical embeddings in $\catw{Cat}$. Therefore, every diagram $\catl{C} \to^D  \catw{ZFA}(\struct{A})$ in  $\catw{ZFA}(\struct{A})$  of the shape of $\catl{C}$ (i.e.~a functor $\catl{C} \to^D  \catw{ZFA}(\struct{A})$) for a \emph{finitely generated} category $\catl{C}$  factors via some embedding $\cont{\aut{\struct{A} \sqcup \word{A_0}}} \rightarrow  \catw{ZFA}(\struct{A})$: 
$$\bfig
\node b1(600, 0)[\cont{\aut{\struct{A} \sqcup \word{A_0}}}]
\node a2(0, 300)[\catl{C}]
\node b2(600, 300)[\catw{ZFA}(\struct{A})]

\arrow|m|/->/[a2`b2;D]

\arrow|r|/-->/[a2`b1;]
\arrow|r|/{^{(}->}/[b1`b2;]

\efig$$

\begin{theorem}[Transfer principle]\label{t:transfer}
Every categorical reasoning concerning finitely generated diagrams consisting of elementary topos operations, such as: finite limits and colimits, exponentials, power objects, quotients, internal quantifiers, etc. can be studied in $\cont{\aut{\struct{A} \sqcup \word{A_0}}}$ for some finite $A_0 \subseteq A$ and then the results can be transferred back to $\catw{ZFA}(\struct{A})$.
\end{theorem}

\begin{corollary}
If $\struct{A}$ is $\omega$-categorical (resp.~extremely amenable) then for every finite $A_0 \subseteq A$, the structure $\struct{A} \sqcup \word{A_0}$ is $\omega$-categorical (resp.~extremely amenable) as well. Therefore, every theorem involving elementary topos construction that holds for every $\omega$-categorical (resp.~extremely amenable) $\struct{A}$ in $\cont{\aut{\struct{A}}}$, also holds in $\catw{ZFA}(\struct{A})$. 
\end{corollary} 

We can slightly extend the above transfer principle by observing that not only the forgetful functor $\catw{Log} \rightarrow \catw{Cat}$ preserves filtered colimits, but also the forgetful functor $\catw{Log} \rightarrow \catw{Cart}$ to the 2-category of cartesian categories and cartesian functors does. Therefore, $\catw{ZFA}(\struct{A})$ is also a filtered colimit of logical embeddings in $\catw{Cart}$. As a consequence, every \emph{cartesian} functor $\catl{C} \to^D  \catw{ZFA}(\struct{A})$ in  $\catw{ZFA}(\struct{A})$ of the shape of a \emph{finitely generated} cartesian category $\catl{C}$ factors via some embedding $\cont{\aut{\struct{A} \sqcup \word{A_0}}} \rightarrow  \catw{ZFA}(\struct{A})$. 

\begin{theorem}[Extended transfer principle]\label{t:extended:transfer}
Let $\mathcal{L}$ be a Lawvere theory with finitely many operations. Then a model of $\mathcal{L}$ in $\catw{ZFA}(\struct{A})$ is a model of $\mathcal{L}$ in $\cont{\aut{\struct{A} \sqcup \word{A_0}}}$ for some finite $A_0 \subseteq A$. Moreover, if $X \in \catw{ZFA}(\struct{A})$ lives in  $\cont{\aut{\struct{A} \sqcup \word{A_0}}}$ then the free $\mathcal{L}$-algebra over $X$ exists in $\catw{ZFA}(\struct{A})$ and is the free $\mathcal{L}$-algebra over $X$ in $\cont{\aut{\struct{A} \sqcup \word{A_0}}}$.
\end{theorem}
\begin{proof}
The only non-trivial part is that $\catw{ZFA}(\struct{A})$ has free $\mathcal{L}$-algebras. For an $A_0$-equivariant set $X \in \catw{ZFA}(\struct{A})$ define its free $\mathcal{L}$-algebra $\mor{X^*}{\mathcal{L}}{\catw{ZFA}(\struct{A})}$ as the free $\mathcal{L}$-algebra $X^*_{A_0}$ on $X$ in $\cont{\aut{\struct{A} \sqcup \word{A_0}}}$ postcomposed with the natural embedding $\cont{\aut{\struct{A} \sqcup \word{A_0}}} \rightarrow \catw{ZFA}(\struct{A})$. Because the logical functors $\cont{\aut{\struct{A} \sqcup \word{A_0}}} \rightarrow \cont{\aut{\struct{A} \sqcup \word{A_1}}}$ are (co)continuous, the definition of $X^*$ does not depend on a particular choice of $A_i \subseteq A_0$.

Consider any $\mathcal{L}$-algebra $\mor{M}{\mathcal{L}}{\catw{ZFA}(\struct{A})}$ and a symmetric function $\mor{f}{X}{|M|}$ to the underlying set $|M|$ of $M$. By the definition of cofiltered limits in $\catw{Cart}$ there is some finite $A_1$ such that both $X^*$ and $M$ factors as $\mathcal{L}$-algebras in $\cont{\aut{\struct{A} \sqcup \word{A_0}}}$ and $\mor{f}{X}{|M|}$ is $A_1$-symmetric. Because $X^*_{A_1}$ is free, function $f_{A_1}$ uniquely extends to a homomorphism $\mor{f^*_{A_1}}{X^*_{A_1}}{M_{A_1}}$, which after embedding into $\catw{ZFA}(\struct{A})$ yields a homomorphism $\mor{f^*}{X^*}{M}$.

For the uniqueness of $f^*$ consider another homomorphism $\mor{g^*}{X^*}{M}$ whose restriction to generators $X$ is equal to $f$. Proceeding like in the above, by the definition of 2-filtered colimits, we may find some finite $A_2$ such that both $f^*$ and $g^*$ factors as homomorphisms $f^*_{A_2}$ and $g^*_{A_2}$ in $\cont{\aut{\struct{A} \sqcup \word{A_2}}}$. Then $g_{A_2} = f_{A_2}$ implies that  $f^*_{A_2} = g^*_{A_2}$ and so $f^* = g^*$.
\end{proof}

Of course, the above can be also generalised to finitely generated cocartesian categories,  finitely generated bicartesian categories,  finitely generated elementary toposes, etc. Extended transfer principle allows us to transfer reasoning about concepts like languages in $\catw{ZFA}(\struct{A})$ to the reasoning in $\cont{\aut{\struct{A} \sqcup \word{A_0}}}$ for some finite $A_0 \subseteq A$.

\begin{remark}
According to Theorem~\ref{t:extended:transfer} every Lawvere theory with finitely many operations has free algebras. This is no longer true for Lawvere theories with infinitely many operations. For a counterexample consider the Lawvere theory consisting of infinitely many constants $(c_i)_{i \in \mathcal{N}}$. By the considerations like in the proof of Theorem~\ref{t:extended:transfer}, if the free algebra over the empty set existed in $\catw{ZFA}(\struct{A})$, it must have been an equivariant infinite set. There is, however, no symmetric function from an equivariant set to the set of atoms $A$ with infinitely many different values, and $A$ may be turned into an algebra by interpreting $(c_i)_{i \in \mathcal{N}}$ as different atoms.   
\end{remark}

\begin{remark}\label{r:zfa:topos}
As a consequence of Lemma~\ref{l:presentation} category $\catw{ZFA}(\struct{A})$ is an elementary topos with the topos structure inherited from $\cont{\aut{\struct{A} \sqcup \word{A_0}}}$. It is not, however, a Grothendieck topos in general: as we shall see in the next subsetion it may lack  many infinite colimits.
\end{remark}

\subsection{First order structures}\label{ss:first:order:structure} 
We shall say that an $A_0$-equivariant set $X \in \catw{ZFA}(\struct{A})$ is of finitary type if its canonical action has only finitely many orbits, i.e.~if the relation $x \equiv y \Leftrightarrow \exists_{\pi \in \aut{\struct{A} \sqcup \word{A_0}}} \; x = \pi \bullet y$ has finitely many equivalence classes. The reason behind this terminology is that $X$ is of a finitary type if it is compact when treated as an object of category $\cont{\aut{\struct{A} \sqcup \word{A_0}}}$. Let us recall the formal definition of a compact object in a general category with filtered colimits. 
\begin{definition}[Compact object]\label{d:compact}
An object $X$ of a category $\catl{C}$ is called compact if its co-representation ${\mor{\hom_{\catl{C}}(X, -)}{\catl{C}}{\catw{Set}}}$ preserves filtered colimits of monomorphisms.
\end{definition}

\begin{example}\label{e:compact}
Here are some examples of compact objects in toposes: 
\begin{itemize}
\item a set $X$ in classical $\catw{Set}$ is compact iff it is finite
\item a continuous $\group{G}$-set $X$ in $\cont{\group{G}}$ is compact iff it has finitely many orbits 
\item a function $\mor{X}{A}{B}$ thought of as an object in $\catw{Set}^{\bullet \rightarrow \bullet}$ is compact if its graph is finite --- i.e.~if $A$ and $B$ are finite sets
\item a chain of functions $(\mor{X_i}{A_i}{A_{i+1}})_{i \in N}$ thought of as an object in $\catw{Set}^{\bullet \rightarrow \bullet \rightarrow \bullet \dotsc}$ is compact if all $A_i$ are finite and the chain is eventually bijective
\end{itemize}
\end{example}

Observe that we cannot speak about compact objects in $\catw{ZFA}(\struct{A})$, because $\catw{ZFA}(\struct{A})$ does not have filtered colimits of monomorphisms. For a counterexample consider the chain:
$$\{\} \subset \{a_1\} \subset \{a_1, a_2\} \subset \{a_1, a_2, a_3\} \subset \cdots$$
where $a_i \in A$. This chain cannot have a colimit, since not every function $f$ with domain $A$ is symmetric, but every restriction of $f$ to a finite set is symmetric. Therefore, by Theorem~\ref{t:transfer}, the notion of a set of finitary type in $\catw{ZFA}(\struct{A})$ is a reflection of the notion of compacteness in $\cont{\aut{\struct{A} \sqcup \word{A_0}}}$ . Every set of finitary type is isomorphic to a set that is hereditarily of finitary type, therefore without loss of generality we can assume that all finitary sets are of this form.  We call a set in $\catw{ZFA}(\struct{A})$ ``definable with atoms'' if it is hereditarily of finitary type. The category of definable sets and functions with atoms will be denoted by $\catw{Def}(\struct{A})$, and its subcategory of $A_0$-equivariant sets by $\catw{Def}_{A_0}(\struct{A})$. 
\begin{theorem}[Presentation of $\catw{Def}(\struct{A})$]
Category $\catw{Def}(\struct{A})$ is the filtered colimit of categories $\catw{Def}_{A_0}(\struct{A})$ and natural embeddings for finite subsets $A_0 \subset A$. 
\end{theorem}   

Blass and Scedrov in \cite{blass1983boolean} proved that $\cont{\aut{\struct{A}}}$ is a coherent topos if and only if $\struct{A}$ is $\omega$-categorical. Moreover, in such a case first order definable subsets\footnote{Definability means here: ``definable in the theory of $\struct{A}$ extended with elimination of imaginaries''. For more details see the next subsection.} of $\struct{A}$ coincide (up to isomorphisms) with compact objects in $\cont{\aut{\struct{A}}}$. The last statement is a very special case of the characterisation theorem for coherent toposes by Alexander Grothendieck and we return to it in Section~\ref{ss:classifying:topos}. We point out, that one direction of this theorem for oligomorphic structures was recently rediscovered in \cite{DBLP:conf/lics/BojanczykKL11}.

Let us recall that by Ryll-Nardzewski theorem \cite{ryll1959categoricity}, a structure $\struct{A}$ (in a countable language) is $\omega$-categorical if and only if for every $k$, there are only finitely many non-equivalent formulas with $k$ free variables. By the above considerations, this can be equivalently expressed by the following property of $\cont{\aut{\struct{A}}}$: every compact object has only finitely many subobjects; or by the property of $\catw{ZFA}(\struct{A})$: every set of finitary type has only finitely many $A_0$-equivariant subsets (for every finite $A_0 \subseteq A$). This property allows for effective algorithms in $\catw{ZFA}(\struct{A})$ on sets of finitary type. This has been observed in \cite{DBLP:conf/lics/BojanczykKL11}. If one is careful to use only elementary topos operations in algorithms then, because every algorithm is finite, it can be executed in $\cont{\aut{\struct{A} \sqcup \word{A_0}}}$ and by transfer principle its outcome transfers to $\catw{ZFA}(\struct{A})$. Notice however, that a power set of a set of finitary type needs not be of a finitary type. Therefore, one should further restrict to the operations that are stable under definability, i.e. one may use: finite limits, finite coproducts and coequalisers of kernel pairs (i.e.~quotient sets), Boolean operations on definable subsets, images, inverse images and dual images of definable sets under definable functions.
    
Let us assume that $\struct{A}$ is $\omega$-categorical with a decidable theory. Algorithm~\ref{a:reachability} can be run on a definable relation $E$ and two elements $a, b \in \catw{Def}(\struct{A})$. Moreover, it always terminates on such inputs. Its run can be seen as a computation of a partial transitive-reflexive closure of a relation $E$ . By transfer principle, we can assume that the inputs are equivariant. Then there is a formula $\phi$ that defines $E$, and all relations: $\id{}, E, E^2, \dotsc$ have the same context as $\phi$. Thus there are only finitely many different $E^k$ and the process terminates after finitely many steps.
We are generally interested in the structures $\struct{A}$ with the property that algorithms like Algorithm~\ref{a:reachability} can be effectively realized. Unfortunately, if $\struct{A}$ is not $\omega$-categorical, then such a problem is not even well-defined, because the correspondence between sets definable in the theory of $\struct{A}$ and sets of finitary types fails badly: there is no longer a correspondence between complete types over $\struct{A}$ and orbits of $\aut{\struct{A}}$; moreover, finite sets of types does not correspond to formulas (for more details consult \cite{hodges1993model} Chapter~10). We shall see later in Section~\ref{ss:classifying:topos} that equivariant definable sets $\catw{Def}_\emptyset(\struct{A})$ can be recovered from the classifying topos of the first order theory of $\struct{A}$ as the full subcategory on a special type of compact objects called \emph{coherent}. Example~\ref{e:compact} shows that compact objects in $\cont{\aut{\struct{A}}}$ can have only finitely many orbits, therefore there cannot be a one-to-one correspondence between compact (nor coherent) objects in $\cont{\aut{\struct{A}}}$ and equivariant definable sets $\catw{Def}_\emptyset(\struct{A})$ for non-$\omega$-categorical $\struct{A}$. In fact, Blass and Scedrov (see Section~\ref{ss:classifying:topos}) showed that the classifying topos for the first order theory of $\struct{A}$ cannot be even Boolean unless $\struct{A}$ is $\omega$-categorical.     

Instead of diving into classifying toposes, which for general structures may be difficult to describe, we may like to reverse our thinking and treat definable algorithms as formulas themselves: in the sense of dynamic logic. Then the question about effective realisation of algorithms turns into the question of decidability of the first order theory extended with dynamic logic of structure $\struct{A}$. Let us focus on the following well-understood fragment of dynamic logic:  $\mu$-calculus --- i.e.~extension of first order logic with the least fixed-point operator. That is, together with the usual first order formulas, we also have formulas of the form:
$\mu X[\overline{y}] . \phi(X, \overline{y})$,
where $X$ (must occur positively in $\phi$) is a ``predicate'' variable of arity equal to the length of sequence of ``parameters'' $\overline{y}$. The semantics of this formula (in a given algebraic structure) is the least set $X^*$ such that: $X^*(\overline{y})  \Leftrightarrow \phi(X^*, \overline{y})$. For example, if $E$ is a formula representing a binary relation, then:
$$\mu X[y_1, y_2] .\; (y_1 = y_2) \vee (\exists_z E(y_1, z) \wedge X(z, y_2))$$
defines the transitive-reflexive closure of $E$.
The least fixed point is the union of the following sequence defined by transfinite induction:
\begin{itemize}
\item $X_0(\overline{y}) = \bot$
\item $X_{\alpha+1}(\overline{y}) =  \phi(X_\alpha, \overline{y})$
\item $X_{\lambda}(\overline{y}) =  \bigvee_{\alpha < \lambda} \phi(X_\alpha, \overline{y})$ if $\lambda$ is a limit ordinal
\end{itemize}
Because the above sets (called the stages of computation) are bounded by the context of $\overline{y}$, the transfinite sequence must stabilize at some ordinal $\alpha$, in which case we put $X^* = X_\alpha$. In particular, if $\struct{A}$ is countable, which is the case of computations in set with atoms (see \cite{DBLP:conf/lics/BojanczykKL11}, \cite{DBLP:conf/popl/BojanczykBKL12}, \cite{DBLP:conf/lics/BojanczykKLT13}), $\alpha$ is bounded by a finite ordinal, and the computation of the least fixed-point can be realized by a standard while-program. Therefore, the least fixed-points are computable by while-programs iff all stages of computations stabilize after finitely many steps. This is the case of an $\omega$-categorical structure $\struct{A}$ (since then, in any context $\overline{y}$ there are only finitely many definable sets). Unfortunately, this is a general phenomenon in case of decidability.

\begin{theorem}[$\mu$-elimination in decidable theories]\label{t:stage}
Let $\struct{A}$ be a first order structure. If the first order theory extended by the least fixed-point operator of $\struct{A}$ is decidable, then every fixed point formula in $\struct{A}$ is first order definable.
\end{theorem}
\begin{proof}
We shall show that if a theory is decidable then there cannot be any formula $\mu X[\overline{y}] . \phi(X, \overline{y})$ such that there is an infinite ascending chain $X_0 \subset X_1 \subset X_2 \subset \cdots$. But in such a case, $X^* = X_k$ for some finite $k$, thus $X^*$ is first order definable.

To complete the proof, let us assume that there is a formula $\mu X[\overline{y}] . \phi(X, \overline{y})$ with an infinite ascending chain $X_0 \subset X_1 \subset X_2 \subset \cdots$. Moschovakis's stage comparison theorem \cite{moschovakis2014elementary} (see also \cite{moschovakis1974nonmonotone} and \cite{kreutzer2004expressive}) says that the relation $\overline{y} \leq \overline{y'} \Leftrightarrow (X_\beta(\overline{y'}) \rightarrow X_\beta(\overline{y}))$ is also definable in $\mu$-calculus. Let us write $\overline{y} \equiv \overline{y'}$ for $\overline{y} \leq \overline{y'} \wedge \overline{y'} \leq \overline{y}$ and $\overline{y} < \overline{y'}$ for $\overline{y} \leq \overline{y'} \wedge \overline{y'} \not\leq \overline{y}$. Then we can define natural numbers up to $\equiv$-equivalence in the following way:
\begin{itemize}
\item $\word{zero}(\overline{x})\colon \forall_{\overline{y}}\; {\overline{y} \leq \overline{x}} \Rightarrow \overline{x} \equiv \overline{y}$
\item $\word{succ}(\overline{x}, \overline{y}) \colon \overline{x} < \overline{y} \wedge \forall_{\overline{z}}\; \overline{x} < \overline{z} \Rightarrow \overline{y} \leq \overline{z}$
\end{itemize}
Because the chain $X_0 \subset X_1 \subset X_2 \subset \cdots$ is infinite, every number has its successor, i.e.: $\forall_{\overline{x}} \exists_{\overline{y}} \colon \word{succ}(\overline{x}, \overline{y})$.
The addition relation on these numbers can be expressed as the following least fixed point:
\begin{multline*}
\word{add}(\overline{x}, \overline{y}, \overline{z}) \colon \\
\mu A[\overline{x},\overline{y},\overline{z}] . (\word{zero}(\overline{y}) \wedge \overline{x} \equiv \overline{z}) \vee \exists_{\overline{p}, \overline{s}} \word{succ}(\overline{p}, \overline{y}) \wedge \word{succ}(\overline{x}, \overline{s}) \wedge A(\overline{s}, \overline{p}, \overline{z})
\end{multline*}
Similarly, the multiplication may be defined by the formula:
\begin{multline*}
\word{mul}(\overline{x}, \overline{y}, \overline{z}) \colon \\
\mu M[\overline{x},\overline{y},\overline{z}] . (\word{zero}(\overline{y}) \wedge \word{zero}(\overline{z})) \vee \exists_{\overline{p}, \overline{s}} \word{succ}(\overline{p}, \overline{y}) \wedge M(\overline{x}, \overline{p}, \overline{s}) \wedge \word{add}(\overline{x}, \overline{s}, \overline{z})
\end{multline*}
What makes the theory undecidable.
\end{proof}

The above theorem says that a decidable theory has to eliminate least fixed-points operators. As mentioned in the above, a practical consequence of this fact is that if $\struct{A}$ is a general countable structure, then if an algorithm over $\struct{A}$ can be effectively realized, then its result must be already present in the first order theory of $\struct{A}$. 
We claim that it is the consequence of Theorem~\ref{t:stage} that is crucial for effective realisation of the algorithms, rather than a much stronger property of $\omega$-categoricity. This claim leads us to the following definition.
\begin{definition}[Locally $\omega$-categorical theory]\label{d:locally:omegacat}
Let $T$ be a (not necessarily complete) first order theory. Then we say that $T$ is locally $\omega$-categorical if every finite set of formulas in $T$ closed under logical connectives of first order logic yields a finite set of formulas up to equivalence in $T$.
\end{definition}
Of course, $\omega$-categorical theories are locally $\omega$-categorical. Theories from Example~\ref{e:pure:sets:constants} and Example~\ref{e:dense:order} are locally $\omega$-categorical, but not $\omega$-categorical. We have the following theorem.
\begin{theorem}\label{t:locally:eliminates}
Let $T$ be a decidable locally $\omega$-categorical theory. Then $T$ eliminates least fixed-points.
\end{theorem}
\begin{proof}
Let $\mu X[\overline{y}] . \phi(X, \overline{y})$ be a formula such that $\phi$ is first order. Then for each finite $i$ The $i$-th stage of computation of $\mu X[\overline{y}] . \phi(X, \overline{y})$ is definable by a first order formula, which is a first order combination of atomic formulas from $\phi$. But there are only finitely many of such formulas by local $\omega$-categoricity of $T$. Therefore, $X_i = X_{i+1} = \mu X[\overline{y}] . \phi(X, \overline{y})$ for some finite $i$.
\end{proof}
In fact, if $T$ is a decidable locally $\omega$-categorical theory, then every while-program terminates on $T$-definable sets. We will elaborate more on $T$-definable sets in non-complete theories in the next subsection.

In the reminder we shall consider positive existential theories that satisfy some weaker versions of the conclusion of Theorem~\ref{t:stage}. The pay-off for such generality is that we lose correspondence between sets with atoms and definable sets --- the role of $\cont{\aut{\struct{A}}}$ will be played by the classifying topos $\catw{Set}(T)$ of a positive existential theory $T$. Note, however, that because of Theorem~\ref{t:transfer} this lost is not that severe. 


\subsection{Positive existential theories}\label{ss:coherent:theory}
If $\struct{A}$ is a single-sorted algebraic structure, then there is a one-to-one correspondence between first order definable subsets of $A^k$ and first order formulas in the language of $\struct{A}$ up to equivalence modulo the theory $\word{Th}(\struct{A})$ of $\struct{A}$. The reason for this is that in a complete theory (and $\word{Th}(\struct{A})$ is clearly complete) two formulas are equivalent iff they have the same interpretation in any model of the theory. Of course, the same is true if move to multi-sorted structures $\struct{A}$, with the obvious correction that we have to consider definable subsets of $\prod_{i \in I_0} A_i$, where $I_0$ is a finite subset of indices of sorts of $\struct{A}$. It is tempting then to extend the notion of definability from structures $\struct{A}$ to non necessarily complete theories $T$ in the following way: we say that the class of formulas $\phi$ up to equivalence modulo $T$ is a $T$-definable set. With one caveat: we are not interested in all first order formulas, but in the formulas from a restricted fragment of \emph{intuitionistic} first order logic, called positive existential logic, which we shall formally define now. Let $(X_i)_{i \in I}$ be a set of variables. Positive existential formulas in variables $(X_i)_{i \in I}$ over signature $\Sigma$ with sorts $(A_i)_{i \in I}$ are defined inductively according to the following rules:
\begin{itemize}
\item $\top, \bot$
\item $R(t_1, t_2, \dotsc, t_k)$ for a relation symbol $R \subseteq {A_{i_1} \times \cdots \times A_{i_k}}$ in $\Sigma$ and terms $t_1 \colon A_{i_1}, \dotsc, t_k \colon A_{i_k}$  in $\Sigma$
\item $t = q$ for terms $t, q$ over the same sort in $\Sigma$
\item $\phi \land \psi$, $\phi \lor \psi$ for formulas $\phi, \psi$
\item $\exists_{x \in X_i} \phi$ for a formula $\phi$ and a variable $x \in X_i$
\end{itemize}
The reason for this restriction is that it gives us much more flexibility in deciding which sets are definable, and which are not. Let us also recall that on the syntactic level one may substitute a classical first order theory with a positive existential theory having the same definable sets. The idea is to introduce two new relational symbol $P_\phi$ and $N_\phi$ for every first order formula $\phi$, then force $P_\phi$ to be equivalent to $\phi$ and $N_\phi$ to its negation. This process is known under the name ``atomisation'', or ``Morleyisation'' (see \cite{hodges1993model} Chapter~2 or \cite{johnstone2003sketches} Chapter~D1.5). Therefore, one may recover the full power of classical first order logic in positive existential logic.

\begin{example}[Infinite decidable objects]\label{e:infinite:decidable}
The first order theory of pure sets from Example~\ref{e:pure:sets} is equivalent to the following positive existential theory, called the theory of infinite decidable objects. The theory is over signature with a single sort $N$ and one binary relation ${\neq} \subseteq N \times N$ and consists of the following axioms:
\begin{itemize}
\item $a \neq b \wedge a = b \vdash \bot$
\item $\vdash a \neq b \vee a = b$
\item $\exists_{x_1} \exists_{x_2} \cdots \exists_{x_n} x_1 \neq x_2 \wedge \cdots \wedge x_i \neq x_j \cdots \wedge x_{n-1} \neq x_n$
\end{itemize}
The first two axioms say that relation ${\neq}$ is complemented by the equality relation ${=}$. The last axiom scheme describes an infinite sequence of axioms, whose $n$-th axiom says that there are at least $n$ different elements.
\end{example}
If $\phi$ is a definable set in the context $\prod_{i \in I_0} A_i$ and $\psi$ is a definable set in the context $\prod_{i \in I_1} B_i$, then a definable function $\mor{f}{\phi}{\psi}$ from $\phi$ to $\psi$ is a definable set in the context $\prod_{i \in I_0} A_i \times \prod_{i \in I_1} B_i$ satisfying the following (positive existential) axioms:
\begin{eqnarray*}
f(\overline{x}, \overline{y}) &\vdash & \phi(\overline{x}) \wedge \psi(\overline{y}) \\
\phi(\overline{x}) &\vdash & \exists_{\overline{y}}\; f(\overline{x}, \overline{y}) \\
f(\overline{x}, \overline{y_1}) \wedge f(\overline{x}, \overline{y_2}) &\vdash & \overline{y_1} = \overline{y_2}
\end{eqnarray*}
\begin{definition}\label{d:definable:sets}
Let $T$ be a positive existential theory. By $\catw{Def}(T)$ we shall denote the category of $T$-definable sets and $T$-definable functions with natural identities and compositions.
\end{definition}
Category $\catw{Def}(T)$ is a coherent category, i.e.~it has finite limits, stable existential quantifiers and stable unions of subobjects. Moreover, in case $T$ is a classical first order theory, $\catw{Def}(T)$ also has stable universal quantifiers and is Boolean (i.e.~it is Boolean Heyting category). 

Unfortunately, $\catw{Def}(T)$ may lack finite disjoint coproducts in general. It has been observed by Makkai and Reyes \cite{makkai2006first} that any positive existential theory $T$ can be extended to a positive existential theory $T^\sqcup$ in such a way that $T$ and $T^\sqcup$ have essentially the same models and $T^\sqcup$-definable sets admit disjoint coproducts. Their construction follows an intuitive idea of ``encoding'' disjoint coproducts directly in the language of the theory. Theory $T^\sqcup$ is obtained from $T$ by extending the signature of $T$  with a new sort $\coprod_{i \in \lambda} A_i$ for every finite cardinal $\lambda$, together with new functional symbols $\mor{\iota_j}{A_j}{\coprod_{i \in \lambda} A_i}$ for every $j \in \lambda$ and introducing  axioms expressing that $\coprod_{i \in \lambda} A_i$ are disjoint coproducts with injections $\iota_j$:
\begin{eqnarray*}
 \iota_i(\overline{x}) = \iota_i(\overline{y}) & \vdash &  \overline{x} = \overline{y}\\
\iota_i(\overline{x}) = \iota_j(\overline{y}) & \vdash  & \bot \,\;\;\;\;\;\;\;\;\;\;\;\;\;\;\;\;\;\;\;\;\;\;\;\;\;\;\;\;\textit{for all $i \neq j$}\\
& \vdash & \bigvee_{i \in \lambda} \exists_{\overline{x} \in A_i} \iota_i(\overline{x}) = \overline{z} \;\;\;\;\;\textit{for $\overline{z} \colon \coprod_{i \in \lambda} A_i$}\\
\end{eqnarray*}
It is routine to check that category $\catw{Def}(T^\sqcup)$ has disjoint coproducts.

There is one more important set theoretic construction that may be missing in $\catw{Def}(T)$. Let $R$ be an equivalence relation on a set $X$. Then, one may form the quotient set $X/R$ of $X$ by $R$:
$$X/R = \{ \tuple{x, \{y \colon R(x, y)\}} \colon x \in X\}$$
Moreover, there is also a canonical surjection $\mor{[-]}{X}{X/R}$ sending an element of $X$ to its abstraction class:
$$[x] = \tuple{x, \{y \colon R(x, y)\}}$$
We call $X/R$ together with surjection $\mor{[-]}{X}{X/R}$ an \emph{effective quotient}. A similar trick to the construction of disjoint coproducts, also works for effective quotients, and in fact is much older. In 1978 Saharon Shelah working on stability theory introduced the concept of elimination of imaginaries \cite{shelah1990classification}. An imaginary element of a theory $T$ is an equivalence class of elements that satisfy a given equivalence formula. A theory $T$ is said to eliminate imaginaries if every imaginary element can be treat as a genuine element of $T$. Saharon Shelah defined a process of associating with a theory $T$ another theory $T^{eq}$ in such a way that both $T$ and $T^{eq}$ have the same models, and $T^{eq}$ eliminates imaginaries. In more details, the language of $T^{eq}$ extends the language of $T$ for every equivalence relation $\phi(\overline{x}, \overline{y})$, where $\overline{x}, \overline{y} \in \prod_i A_i$ by a new sort $A_\phi$ together with a new functional symbol $\mor{e_\phi}{\prod_i A_i}{A_\phi}$.
The theory $T^{eq}$ consists of the axioms of $T$, plus additional axioms expressing that $e_\phi$ is a surjective function onto the set of equivalence classes of $\phi$:
\begin{eqnarray*}
&\vdash & \exists_{\overline{x}} e_\phi(\overline{x}) = \overline{y}\\
\phi(\overline{x}, \overline{y}) &\vdash & e_\phi(\overline{x}) = e_\phi(\overline{y})\\
e_\phi(\overline{x}) = e_\phi(\overline{y}) &\vdash & \phi(\overline{x}, \overline{y})\\
\end{eqnarray*}
Again, it is routine to check that category $\catw{Def}(T^{eq})$ has effective quotients. Moreover, $(-)^{eq}$ preserves $(-)^{\sqcup}$ construction, therefore $\catw{Def}(T^+)$ for $T^+ = {T^{\sqcup}}^{eq}$ has disjoint coproducts and effective quotients. The process of constructing $\catw{Def}(T^+)$ from $\catw{Def}(T)$ is known in category theory under the name \emph{pretopos completion}. Furthermore, $T^+$ is the maximal tight extension of $T$, therefore we cannot further extend $T^+$ without changing the notion of a model.

If models of a first order theory $T$ satisfy the following properties: (a) every sort of every model is non-empty, (b) there is at least one sort that has at least two elements in every model; then the category of $T^{eq}$-definable objects $\catw{Def}(T^{eq})$ automatically has disjoint coproducts. In particular, observe that all non-trivial $\omega$-categorical theories are essentially of this type.

\begin{example}[Definablility in an $\omega$-categorical theory]
If $T$ is a classical $\omega$-categorical theory with an infinite model $\struct{A}$, then $\catw{Def}(T^+)$ is equivalent to $\catw{Def}_\emptyset(\struct{A})$. The reason for this equivalence is that sets in $\catw{Def}_\emptyset(\struct{A})$ can be \emph{nested}, i.e.~one may form sets like:
\begin{align*}
&\{\{a, b\} \colon a, b \in A \} \\
&\{\tuple{a, \{b \in A \colon \phi(a, b)\}} \colon a \in A \wedge \psi(a) \}
\end{align*}
If we allow for nested-set definitions, then we may forget about supplying $T^\sqcup$ with imaginary elements. This is because, the set of quotients of an equivalence relation $\phi(\overline{x}, \overline{y})$ on $X$ may be represented as:
\begin{align*}
\{\tuple{\overline{x}, \{\overline{y} \in X \colon \phi(\overline{x}, \overline{y})\}} \colon \overline{x} \in X \}
\end{align*}
Conversely, if $\phi(\overline{x}, \overline{y})$ is \emph{any} formula, then one may form an equivalence formula: $$\widehat{\phi}(\overline{x}, \overline{x}') = \forall_{\overline{y}} \phi(\overline{x}, \overline{y}) \leftrightarrow \phi(\overline{x}', \overline{y})$$ and represent $\{\overline{y} \colon \phi(\overline{x}, \overline{y})\}$ by an imaginary element of $\widehat{\phi}(\overline{x}, \overline{x}')$ of $T^+$. Then by induction over structure of nested-sets one can show that every nested-definable set is $T^+$-definable.

The above is also a consequence of the characterisation theorem for Boolean coherent toposes discussed in the next subsection. 
\end{example}

\subsection{Classifying topos}\label{ss:classifying:topos}
One important connection between Grothendieck toposes and logic is through the concept of classification. The general statement says that, for every logical theory $T$ formalized in a positive existential fragment of infinitary first order logic there is a Grothendieck topos $\catw{Set}(T)$, and a generic model $\struct{M}_T$ of $T$ inside $\catw{Set}(T)$. The term \emph{generic} means that \emph{every} model of the theory in \emph{any} Grothendieck topos can be obtained from $\struct{M}_T$ as an application of the inverse image part of some geometric morphism into $\catw{Set}(T)$. Roughly speaking, a generic model of a theory is a well-behaved model that contains all of the information about the theory. The topos $\catw{Set}(T)$ is called the classifying topos for $T$. Moreover, the above general statement is definitive, because every Grothendieck topos arises as the classifying topos for some positive existential fragment of infinitary first-order theory.


In this paper we are interested in theories $T$ defined in positive existential fragment of \emph{finitary} first order logic. In such a case, the classifying topos $\catw{Set}(T)$ is called \emph{coherent topos} and can be obtained as the topos of sheaves on $\catw{Def}(T)$ with the usual coherent topology (i.e.~topology generated by finite jointly regular-epimorphic families of morphisms; for more information consult \cite{johnstone2003sketches} Chapter~D3, especially Section~D3.3, or Volume~3 of \cite{borceux}) Moreover, the categories of sheaves on $\catw{Def}(T)$ and on $\catw{Def}(T^+)$ are equivalent, i.e.~$\catw{Set}(T) \approx \catw{Set}(T^+)$.
\begin{example}[Sierpienski topos]\label{e:sierpienski:topos}
The classifying topos of the propositional theory from Example~\ref{e:prop:theory} is the Sierpienski topos $\catw{Set}^{\bullet \rightarrow \bullet}$ --- i.e.~the topos of sheaves on the Sierpienski space $\Sigma$. To see this from the perspective of the definable sets, observe that there are exactly three definable sets in this theory: corresponding to the false formula $\bot$, to nullary relation $p$ itself and to the true formula $\top$. Moreover, the coherent topology on the definable sets is the same as the topology of the Sierpienski space $\{0, 1\}$, whose open sets are $\bot = \emptyset$, $p = \{0\}$ and $\top = \{0, 1\}$.
\end{example}

\begin{example}[Impossible topos]\label{e:impossible:topos}
The classifying topos of the seemingly impossible theory of Example~\ref{e:impossible:theory} is the presheaf topos $\catw{Set}^{\bullet \rightarrow \bullet \rightarrow \bullet \dotsc}$. 
Like in Example~\ref{e:sierpienski:topos} this topos can be presented as a topos of sheaves on a suitable topological space.
\end{example}
Let us recall the definition of a \emph{coherent object}.
\begin{definition}[Coherent object]
An object $A$ in a category with filtered colimits and kernel pairs is coherent if it is compact and for every morphisms $\mor{f}{B}{A}$ from a compact object $B$ the kernel $\word{Ker}(f)$ of $f$ is compact.  
\end{definition}
It is a classical result of Alexander Grothendieck that $\catw{Def}(T^+)$ can be recovered from the classifying topos $\catw{Set}(T)$ as the full subcategory spanned on coherent objects (see Corollary~3.3.8 in Chapter~D of \cite{johnstone2003sketches}).

\begin{example}\label{e:coherent}
Continuing Example~\ref{e:compact}: all compact objects in $\catw{Set}$, $\catw{Set}^{\bullet \rightarrow \bullet}$ and $\catw{Set}^{\bullet \rightarrow \bullet \rightarrow \bullet \dotsc}$ are coherent. More generally, in a coherent topos, compact objects coincide with coherent objects iff every sub-compact object is compact. Therefore, they coincide in $\cont{\group{G}}$  iff $\group{G}$ is (equivalent to) a group of automorphism of an $\omega$-categorical structure. For a counterexample, consider $\catw{Set}^{\group{G}} = \cont{\group{G}}$ for a discrete group $\group{G}$. Coherent objects in $\catw{Set}^{\group{G}}$ are these sets $X$ with finitely many orbits whose stabilisers $\{\pi \in \group{G} \colon \pi \bullet x = x \}$ are finite at every $x \in X$.  
\end{example}

During 1974-1975 Walter Roelcke in a course on topology at the University of Munich introduced and systematically developed the theory of four natural uniform structures (or uniformities) on topological groups \cite{trove.nla.gov.au/work/25027833}. The lower (infinium) uniformity plays a crucial role in model theory and is nowadays known as Roelcke uniformity. Formally, 
%
let $\group{G}$ be a topological group. The Roelcke uniformity on $\group{G}$ is the uniformity generated by entourages of the form $\{\tuple{g, \alpha \bullet g \bullet \beta} \colon g \in G, \alpha, \beta \in U\}$ for some open $U \subseteq G$ containing the neutral element of $\struct{G}$ (i.e.~the neighbourhood of the identity) such that $\{g^{-1} \colon g \in U\} = U$. 

\begin{remark}
The topology generated by Roelcke uniformity on $\group{G}$ coincides with the topology of $\group{G}$. Moreover, all group operations become uniformly continuous with respect to the Roelcke uniformity. 
\end{remark}

A topological group whose Roelcke uniformity is precompact (its Cauchy-completion is compact) is called \emph{Roelcke precompact}. Explicitly, we have the following definition. A topological group $\group{G}$ is Roelcke precompact if for every entourage $E$ in its Roelcke uniformity there exists a finite set $G_0 \subseteq G$ such that $E[G_0] = G$.

An important characterisation theorem of Roelcke precompact groups is given in \cite{Tsankov2012} as Theorem~2.4.: a topological subgroup $\group{G} \leq \group{S_\infty}$ (i.e.~a non-Archimedean group) is Roelcke precompact iff for every continuous action $\group{G}$ on a countable, discrete set $X$ with finitely many orbits, the induced action on $X^n$ has finitely many orbits for each natural $n$. This theorem says that Roelcke precompact groups are generalizations of oligomorphic groups, capturing their most important properties. In fact, Roelcke precompact groups are multi-sorted metric analogue of oligomorphic groups form the classical model theory \cite{yaacov2006model} \cite{ben2016weakly}.

In the early '80s Andreas Blass and Andre Scedrov \cite{blass1983boolean} tuning the representation theorem of Andr\'{e} Joyal and Myles Tierney \cite{joyal/tierney:1984} (see also \cite{butz1998representing}) to Boolean toposes introduced the notion of a \emph{coherent group}. A topological group $\group{G}$ is coherent if the topos of its continuous actions $\cont{\group{G}}$ is a coherent topos. From this definition, they obtained the following characterisation: a topological group $G$ is coherent if its every open subgroup $\group{H} \subseteq \group{G}$ has only finitely many double cosets: $\group{H}x\group{H} = \{h\bullet x \bullet k \colon h \in H, k \in H\}$ and $x \in G$. Independently, this property, has been also shown to characterise Roelcke precompact groups with small open subgroups. Therefore, Roelcke precompact groups and coherent groups coincide (with small open subgroups)\footnote{To the best knowledge of the authors, this coincidence has not been observed before and the connection between
metric model theory and classifying toposes has never been exploited.}.

The abovementioned theorem states that \emph{Boolean} coherent Grothendieck toposes are precisely the finite products of categories of actions of topological groups of automorphisms of multi-sorted $\omega$-categorical structures \cite{blass1983boolean} (this justifies the correspondence between nominal sets and classifying toposes from Figure~\ref{f:topos:theory}). Moreover, positive existential theories $T$ that are classified by Boolean Grothendieck toposes are characterised by the following properties: a) in every context there are only finitely many formulas up to equivalence modulo $T$, b) every first order formula is classically equivalent to a coherent formula modulo the theory \cite{blass1983boolean}.

Such theories $T$ have only finitely many completions, all of which are $\omega$-categorical. In fact, for a complete theory, the first property is equivalent to $\omega$-categoricity of the theory by (generalised) Ryll-Nardzewski theorem. The classifying topos may be constructed as the product of categories of the form $\cont{\aut{\struct{M}_\word{i}}}$, where $\aut{\struct{M}_i}$ is the topological group of automorphism of the unique countable model $\struct{M}_i$ corresponding to the $i$-th completion of theory $T$ \cite{blass1983boolean}. The theory of dense linear orders from Example~\ref{e:dense:order} satisfies these properties. Moreover, all theories that satisfy these properties are locally $\omega$-categorical.

\begin{example}[Topos for dense linear orders]\label{e:dense:topos}
The theory $T$ of dense linear orders has five completions corresponding to the following models:
\begin{itemize}
\item the one-element structure $1$
\item the structure of rational numbers with the natural ordering $\struct{Q}$
\item the structure of rational numbers extended with plus infinity with the natural ordering $\struct{Q^*}$
\item the structure of rational numbers extended with minus infinity with the natural ordering $\struct{Q_*}$
\item the structure of rational numbers extended with both plus and minus infinity with the natural ordering $\struct{Q^*_*}$
\end{itemize}
Because the theories of these models are $\omega$-categorical, the classifying topos $\catw{Set}(T)$ can be constructed as the product of categories: $$\catw{Set} \times \cont{\aut{\struct{Q}}} \times \cont{\aut{\struct{Q^*}}} \times \cont{\aut{\struct{Q_*}}} \times \cont{\aut{\struct{Q^*_*}}}$$ 
\end{example}

As mentioned in the introduction, we are interested in theories $T$ such that Algorithm~\ref{a:reachability} is effective on $T$-definable sets. One property of such theories is local $\omega$-categoricity introduced in Section~\ref{d:locally:omegacat} for classical first order theories. We could define a suitable version of local $\omega$-categoricity for positive existential theories, but for the purpose of this paper it suffices to require a weaker property. Let us recall that by Theorem~\ref{t:locally:eliminates}, every locally $\omega$-categorical theory eliminates least fixed points. Moreover, this elimination is finitistic in the sense that the natural iterative algorithm to compute the least fixed point terminates after finitely many steps. Here, we shall focus on a weaker property --- i.e.~elimination of transitive closures.
\begin{definition}[Elimination of transitive closures]
Let $T$ be a positive existential theory. We say that $T$ eliminates transitive closures if for every $T$-definable binary relation $R$ there is a finite $n$ such that $R^+ = \bigvee_{i=1}^n R^i$ is transitive.
\end{definition}
Observe that in such a case $R^+$ is the transitive closure of $R$ --- i.e.~if $S \supseteq R$ is any transitive relation, then $R^+ = \bigvee_{i=1}^n R^i \subseteq \bigvee_{i=1}^n S^i = \bigvee_{i=1}^n S = S$.
\begin{example}
All working examples of theories in this paper eliminate transitive closures:
\begin{itemize}
\item if $T$ is $\omega$-categorical then for every $T$-definable binary relation $R$, there are only finitely many $T$-definable relations in the same context; therefore $\bigvee_i R^i$ may be reduced to a finite disjunction; this includes Example~\ref{e:pure:sets}, Example~\ref{e:rationals} and Example~\ref{e:multi:omegacat}
\item more generally, if $T$ is locally $\omega$-categorical then since $T$ eliminates least fixed points, it also eliminates transitive closures; this includes Example~\ref{e:pure:sets:constants} and Example~\ref{e:dense:order} 
\item if $T$ is the propositional theory from Example~\ref{e:prop:theory} or the seemingly impossible theory from Example~\ref{e:impossible:theory}, then every infinite ascending chain of $T^+$-definable sets $\phi_0 \subseteq \phi_1 \subseteq \phi_2 \subseteq \cdots$ stabilizes, therefore $\bigvee_i R^i = \bigvee_{i \leq k} R^i$ for some finite $k$
\end{itemize}
\end{example}

\begin{theorem}\label{t:effective:reachability}
Let $T$ be a decidable positive existential theory that eliminates transitive closures. Then Algorithm~\ref{a:reachability} is effective on $T$-definable inputs.
\end{theorem}

\section{Automata}
\label{s:automata}
For this section we fix a single positive existential  theory $T$ that has disjoint coproducts and eliminates imaginaries. Moreover, if not stated otherwise, all objects and morphisms live in $\catw{Set}(T)$ --- the classifying topos of $T$. The subobject classifier will be denoted by $\Omega$, and the characteristic function of a subobject $\mor{s}{A_0}{A}$, by $\mor{\chi_{s}}{A}{\Omega}$. 
\subsection{Preliminaries}
\label{ss:preliminaries}
An object of words in an alphabet $\Sigma$ is the free monoid $\tuple{\Sigma^*, \word{concat}, \epsilon}$ over $\Sigma$ generators. It consists of the concatenation morphism $\mor{\word{concat}}{\Sigma^* \times \Sigma^*}{\Sigma^*}$ and an empty word $\mor{\epsilon}{1}{\Sigma^*}$. By a language $L$ over alphabet $\Sigma$, we shall mean a subobject of $\Sigma^*$, or equivalently a morphism $\mor{L}{\Sigma^*}{\Omega}$. 
\begin{definition}[Automaton]
A non-deterministic automaton $A$ is a quadruple $A = \tuple{\mor{s_0}{I}{S}, \mor{s_f}{F}{S}, \dist{\sigma}{\Sigma \times S}{S}}$, where:
\begin{itemize}
\item $\mor{s_0}{I}{S}$ is a monomorphism of initial states
\item $\mor{s_f}{F}{S}$ is a monomorphism of final states
\item $\dist{\sigma}{\Sigma \times S}{S}$ is the transition relation
\end{itemize}
Automaton $A$ is called deterministic if $I$ is the terminal object $1$ and the transition relation $\sigma$ is functional.
\end{definition}
An automaton $A$ is called a $T$-automaton if all the data from its definition are $T$-definable. Observe that for any object $S$, we have the canonical monoidal structure on $\Omega^{S\times S}$, given by the internal composition of binary relations. Therefore, the adjoint transposition ${\mor{\sigma^\dag}{\Sigma}{\Omega^{S\times S}}}$ induces a unique homomorphisms of monoids: ${\mor{\overline{\sigma^\dag}}{\Sigma^*}{\Omega^{S\times S}}}$, and by transposition a relation: ${\dist{\overline{\sigma}}{\Sigma^* \times S}{S}}$.

\begin{definition}[Language of an automaton]
Let $A = \tuple{\mor{s_0}{I}{S}, \mor{s_f}{F}{S}, \dist{\sigma}{\Sigma \times S}{S}}$ be a non-deterministic automaton. By the language $L(A)$ recognized by $A$ we shall mean the subobject of $\Sigma^*$ that corresponds to the following relation:
$$\Sigma^* \overset{\id{\Sigma^*} \times \chi^\dag_{s_0}}{\slashedrightarrow} \Sigma^* \times S \overset{\overline{\sigma}}\slashedrightarrow S \overset{\chi_{s_f}}\slashedrightarrow 1$$  
\end{definition}
In case automaton $A$ is deterministic, the monoid of relations $\Omega^{S\times S}$ can be substituted by the monoid of functions $S^S$ with internal composition, and $L(A)$ can be constructed as the pullback of $s_f$ along $\mor{\overline{\sigma}^\dag \circ (\id{\Sigma^*} \times s_0)}{\Sigma^*}{S}$.

Before moving to more advanced theory, let us prove a simple theorem. 
\begin{theorem}\label{t:effective:emptiness}
Let $T$ be a decidable theory that eliminates transitive closures of binary relations. Then the problem of emptiness for a $T$-definable automata is decidable.
\end{theorem}
\begin{proof}
The problem of emptiness of an automaton is equivalent to the problem of reachability of a final state from one of the initial states. Therefore it suffices to compute the transitive-reflexive closure $\phi$ of its underlying transition relation by Algorithm~\ref{a:reachability} and then check if the formula $\exists_{s \in s_0} \exists_{f \in s_f} \phi(s, f)$ is provable in $T$.
\end{proof}

Kaminski and Francez studied, so called, finite memory automata \cite{kaminski1994finite}: i.e.~automata augmented with a finite set of registers, each of which can hold a natural number, and the automata can test for equality between registers and alphabets. Here is a suitable generalisation of this definition to a general structure $\struct{A}$.

\begin{definition}[Register automata]\label{d:register:machine}
An $\struct{A}$-automata with $k$ registers over alphabet $\Sigma$ is a quadruple $\tuple{S, \delta, I, F}$ such that:
\begin{itemize}
\item $S$ is a finite set of states
\item $I \subseteq S$ is a set of initial states, and $\phi_I \subseteq A^k$ is a set of possible initial configurations of registers
\item $F \subseteq S$ is a set of final states, and $\phi_F \subseteq A^k$ is a set of possible final configurations of registers
\item $\delta \subseteq (\Sigma \times S \times A^k) \times (S \times A^k)$ is a transition relation such that for every $s, s' \in S$ the relation $\delta(s, s') \subseteq (\Sigma \times A^k) \times A^k$ is $\struct{A}$-definable.
\end{itemize}
\end{definition}
We could state an even more general definition suitable for any theory $T$, but we refrain from doing this for the following reason: every register $\struct{A}$-automata with states $S$ and $k$-registers is equivalent to:
\begin{enumerate}
\item A register $\struct{A}$-automaton with a single state.
\item An $\struct{A}$-automaton (without registers).
\end{enumerate} 
Because $\catw{Def}(\struct{A})$ has disjoint coproducts, it can interpret finite cardinals. Moreover, every function between finite cardinals is definable. Let us assume that the context of $S$ is $A^m$.  Therefore, $S' = S \times A^k$ can be thought of as either the object of states of a definable automaton, or as the $\struct{A}$-automata with $k+m$ registers  and a single state $1$.


\subsection{Myhill-Nerode theorem}\label{ss:myhill}
Consider the following morphism:
$$\Sigma^* \times \Sigma^* \to^{\word{concat}} \Sigma^* \to^{L} \Omega$$
Its transposition $\mor{(L \circ \word{concat})^\dag}{\Sigma^*}{\Omega^{\Sigma^*}}$ maps a word $w$ to the predicate: $\lambda x . wx \in L$.
\begin{definition}[Myhill-Nerode relation]
Let $\mor{L}{\Sigma^*}{\Omega}$ be a language. By the Myhill-Nerode relation $\word{MN}(L)$ of $L$, we shall mean the kernel relation of $(L \circ \word{concat})^\dag$, and by the the Myhill-Nerode quotient of $L$ we shall mean the coequaliser of this kernel pair:
$$\word{MN}(L) \rightrightarrows \Sigma^* \to^{(L \circ \word{concat})^\dag} {\Omega^{\Sigma^*}} $$
\end{definition}
Intuitively, two words $w, v \in \Sigma^*$ are related by Myhill-Nerode relation $\word{MN}(L)$ iff for every $x \in \Sigma^*$ we have that: $wx \in L \Leftrightarrow vx \in L$. 

\begin{lemma}
Let $A = \tuple{\mor{s_0}{1}{S}, \mor{F}{K}{S}, \mor{\sigma}{\Sigma \times S}{S}}$ be a deterministic automaton. The Myhill-Nerode quotient of $L(A)$ is a sub-quotient of $S$.
\end{lemma}
\begin{proof}
We have the following morphism:
$$\bfig
\node a1(0, 400)[\Sigma^* \times \Sigma^*]
\node b1(500, 400)[\Sigma^*]
\node a2(0, 0)[S^S \times S^S]
\node b2(500, 0)[S^S]

\arrow|r|/->/[a1`b1;\mathit{concat}]
\arrow|l|/->/[a2`b2;\circ]

\arrow|m|/->/[a1`a2;\overline{\sigma} \times \overline{\sigma}]
\arrow|m|/->/[b1`b2;\overline{\sigma}]

\node c1(1400, 400)[\Sigma^* \times S]
\node x(1000, 0)[S^S \times S]
\node c2(1400, 0)[S]

\node f(1700, 0)[\Omega]

\arrow|l|/->/[b1`c1;\id{\Sigma^*} \times s_0]
\arrow|r|/->/[b2`x;\id{S^S} \times s_0]
\arrow|r|/->/[x`c2;\mathit{ev} ]
\arrow|r|/->/[c1`c2;\overline{\sigma}^\dag]

\arrow|m|/->/[c1`x;\overline{\sigma} \times \id{S}]

\arrow|r|/->/[c2`f;F]

\efig$$
which by transposition corresponds to the morphism:
$$\mor{k}{\Sigma^*}{\Omega^{\Sigma^*}}$$
The kernel pair $\word{Ker}(k) \rightrightarrows \Sigma^*$ of this morphism $k$, is the Myhill-Nerode relation of the language $L(A)$. Such an equivalence relation induces a quotient object $\Sigma^*/k$ as the coequaliser of the kernel pair:
$$\word{Ker}(k) \rightrightarrows \Sigma^* \to^{[-]_k} \Sigma^*/k $$

On the other hand, the morphism $\mor{\overline{\sigma}^\dag \circ (\id{\Sigma^*} \times s_0)}{\Sigma^*}{S}$, which will be denoted by $s$, has its own kernel pair:
$$\word{Ker}(s) \underset{\pi_2}{\overset{\pi_1}{\rightrightarrows}} \Sigma^* \to^{s} S$$
We want to show that $\Sigma^*/k$ is a quotient of $\Sigma^*/s$, or equivalently that the relation $\word{Ker}(s)$ is coarser than $\word{Ker}(k)$. We shall prove it on generalised elements: $\mor{x,y}{X}{\Sigma^*}$. That is, we want to show that if $s \circ x = s \circ y$ then: $k \circ x = k \circ y$. But, by the triangle equality for exponent: $k \circ x = k \circ y$ iff $k^\dag \circ (x \times \id{\Sigma^*}) =  k^\dag \circ (y\times \id{\Sigma^*})$. Moreover, because $\epsilon$ is the unit for $\word{concat}$, the following diagram commutes:
$$\bfig
\node a1(0, 0)[\Sigma^* \times \Sigma^*]
\node b1(500, 0)[\Sigma^*]
\node a2(0, 400)[X \times \Sigma^*]
\node b2(500, 400)[X]

\arrow|r|/<-/[a2`b2;\tuple{\id{X}, \epsilon}]
\arrow|l|/->/[a1`b1;\mathit{concat}]

\arrow|m|/->/[a2`a1;x \times \id{\Sigma^*}]
\arrow|m|/->/[b2`b1;x]

\efig$$
with the top arrow being mono. Therefore, $k^\dag \circ (x \times \id{\Sigma^*}) =  k^\dag \circ (y\times \id{\Sigma^*})$ iff $k^\dag \circ (x \times \id{\Sigma^*}) \circ  \tuple{\id{X}, \epsilon} =  k^\dag \circ (y\times \id{\Sigma^*}) \circ  \tuple{\id{X}, \epsilon}$ iff $F \circ s \circ x = F \circ s \circ y$, what completes the proof of the claim. Now, because $s \circ \pi^s_1 = s \circ \pi^s_2$, we have that: $k \circ \pi^s_1 = k \circ \pi^s_2$ and by the definition of the kernel  of $k$ there is a unique monomorphism of relations $\mor{j}{\word{Ker}(s)}{\word{Ker}(k)}$, i.e.:~$j \circ \pi^k_1 = \pi^s_1$ and $j \circ \pi^k_2 = \pi^s_2$. Therefore, by the universal property of the coequaliser $\Sigma^*/s$, there is a unique (necessarily epi) morphism: $\Sigma^*/s \rightarrow \Sigma^*/k$. 
\end{proof}

\begin{lemma}
Let $\mor{L}{\Sigma^*}{\Omega}$ be a language. The Myhill-Nerode quotient of $L$ can be equipped with the structure of a deterministic automaton that recognizes $L$.
\end{lemma}


\begin{theorem}[Subcompact rational languages]\label{t:subcompact}
Sub-compact deterministic automata recognize the languages whose Myhill-Nerode quotients are sub-compact.
\end{theorem}
\begin{proof}
If $X$ is a sub-quotient of $A$ via $A_0$, then we may form the pushout:
$$\bfig
\node a1(0, 0)[A_0]
\node b1(300, 0)[X]
\node a2(0, 300)[A]
\node b2(300, 300)[P]

\arrow|r|/->>/[a2`b2;]
\arrow|l|/->>/[a1`b1;e]

\arrow|r|/{^{(}->}/[a1`a2;m]
\arrow|r|/{^{(}->}/[b1`b2;]

\efig$$
A pushout of an epimorphism $\mor{e}{A_0}{X}$ is an epimorphism, thus $P$ is a quotient of $A$. But a quotient of a compact object is compact, so $P$ is compact if $A$ is. Moreover, in a topos a pushout of a monomorphism is again monomorphism. Therefore, if $X$ is a sub-quotient of a compact object, then it is actually a quotient of a compact object.
\end{proof}

From the above theorem we can instantly get the generalisation of Myhill-Nerode Theorem for nominal sets. 
\begin{corollary}\label{c:compact:mn}
In a topos of continuous actions of a topological group, deterministic automata with finitely many orbits recognize exactly the languages whose Myhill-Nerode relations have finitely many orbits.
\end{corollary}

\begin{theorem}\label{t:definable:mn}
Let $T$ be a theory that eliminates transitive closures of binary relations. Then $T$-definable deterministic automata recognize exactly the languages whose Myhill-Nerode quotients are $T$-definable.
\end{theorem}
\begin{proof}
A definable morphism $\mor{\sigma}{\Sigma \times S}{S}$ induces a binary relation on $S$: $R_\sigma(a, b) \leftrightarrow \exists_{x \in \Sigma} \colon \sigma(x, a) = b$. Since $T$ eliminates transitive closures, the transitive closure $R^*_\sigma$ of $R_\sigma$ factors as: $R_\sigma \cup R^2_\sigma \cup \dotsc \cup R^n_\sigma$ for some finite $n$. Unwinding the definition of $R^k_\sigma$, this yields: $R^k_\sigma(a, b) \leftrightarrow \exists_{w \in \Sigma^k} \overline{\sigma}(w,a)=b$. Therefore, $R^*_\sigma(a, b) \leftrightarrow \exists_{w \in \Sigma^{*n}} \; \overline{\sigma}(w,a)=b$, which means that the image of $\mor{\overline{\sigma}}{\Sigma^* \times S}{S}$ factors through $\Sigma^{*n} = \bigsqcup_{i \leq n} \Sigma^i$. This means that $\Sigma^*/s$ is coherent. On the other hand, $\Sigma^*/k$ can be described as the filtered colimit of $\Sigma^{*j}/{k \circ j}$, where $\mor{j}{\Sigma^{*j}}{\Sigma^*}$ is the natural injection of coproducts. Therefore, the epimorphism $\Sigma^*/s \rightarrow \Sigma^*/k$ factors as an epimorphism $\Sigma^*/s \rightarrow \Sigma^{*j}/{{k \circ j}}$ followed by a monomorphism $\Sigma^{*j}/{{k \circ j}} \rightarrow \Sigma^{*}/{k}$. By the uniqueness of epi-mono factorisation, $\Sigma^{*j}/{{k \circ j}} \approx \Sigma^{*}/{k}$, and $\Sigma^{*}/{k}$ is coherent.
\end{proof}

\begin{example}
In all of the theories from Examples~\ref{e:pure:sets}, \ref{e:rationals}, \ref{e:multi:omegacat}, \ref{e:pure:sets:constants}, \ref{e:dense:order},  \ref{e:prop:theory} and \ref{e:impossible:theory} definable deterministic automata recognize exactly the languages whose Myhill-Nerode quotients are definable.
\end{example}

Definable non-deterministic automata are generally more expressive than definable deterministic automata. The reason is that, unlike finite sets, definable sets are not stable under the power-set construction.

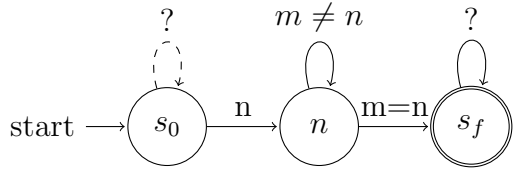
\begin{figure}[tb]
\centering
\begin{tikzpicture}[shorten >=1pt,node distance=2cm,on grid,auto] 
   \node[state,initial] (s_0)   {$s_0$}; 
   \node[state] (n) [right=of s_0] {$n$};
   \node[state,accepting] (s_f) [right=of n] {$s_f$}; 
    \path[->] 
    (s_0) edge  node {n} (n)
    	  edge [dashed,loop above] node {$?$} ()
    (n)   edge  node  {m=n} (s_f)
          edge [loop above] node {${m\neq n}$} ()
    (s_f) edge [loop above] node {$?$} ();
\end{tikzpicture}
\caption{A non-deterministic automaton for language $\{w'' n w' n w \colon n \in \Sigma \wedge w,w',w'' \in \Sigma^*\}$. If we remove the dashed transition, then we obtain a deterministic automaton for language $\{n w' n w \colon n \in \Sigma \wedge w,w' \in \Sigma^*\}$.} 
\label{f:nondet} 
\end{figure}

\begin{example}[Definable deterministic vs.~non-deterministic automata]
Consider the following language in $\catw{Set}(\struct{N})$:
\begin{itemize}
\item the alphabet is the set of all atoms, i.e.: $\Sigma = N$
\item the language consists of all words over alphabet $\Sigma$, such that in each word there is a letter that appears at least twice, i.e.: $L = \{w'' n w' n w \colon n \in \Sigma \wedge w, w', w'' \in \Sigma^*\}$
\end{itemize}
One may check that the Myhill-Nerode quotient of $L$ has infinitely many orbits, therefore $L$ cannot be recognised by a deterministic automaton. On the other hand, the non-deterministic automaton from Figure~\ref{f:nondet} recognizes it: the automaton loops in state $s_0$ for a number of times, non-deterministically moves to the state ``n'' after seeing a letter $n \in \Sigma$, and then loops in that state until another letter $n$ appears in the word, in which case the automaton moves to the final state $s_f$.
\end{example}

\subsection{Recognition by monoids}\label{ss:monoids}
We say that a language $L$ over alphabet $\Sigma$ is recognized by a monoid $\struct{M}$ if there is a subobject $F$ of $M$ and a homomorphism $\mor{\overline{h}}{\Sigma^*}{\struct{M}}$ such that: $\mor{\chi_F \circ \overline{h}}{\Sigma^*}{\Omega}$ is the characteristic morphism of $L$. It is well-known that classical regular languages (i.e.~languages recognised by finite automata in $\catw{Set}$) are precisely the languages recognised by finite monoids. The correspondence does not carry over to definable regular languages and definable monoids --- in general the notion of a language recognised by a coherent deterministic automaton is much stronger than the notion of a language recognised by a coherent monoid.
\begin{example}[Definable deterministic automata vs.~definable monoids]\label{e:det:vs:coh}
Consider the following language in $\catw{Set}(\struct{N})$:
\begin{itemize}
\item the alphabet is the set of all atoms, i.e.: $\Sigma = N$
\item the language consists of all words over alphabet $\Sigma$, such that in each word the first appears at last twice, i.e.: $L = \{n w' n w'' \colon n \in \Sigma \wedge w,w' \in \Sigma^*\}$
\end{itemize} 
One may check that $L$ cannot be recognised by a monoid that has only finitely many orbits. On the other hand, the deterministic part (without the dashed transition) of the automaton from Figure~\ref{f:nondet} clearly recognizes it: the automaton moves to the state $n$ after seeing a letter $n \in \Sigma$, and then loops in that state until another letter $n$ appears in the word, in which case the automaton moves to the final state $s_f$.
\end{example}
Therefore, to hope for such a correspondence, we need a more general notion of a monoid, or a more restrictive notion of an automaton. Languages recognized by finitary monoids in $\catw{ZFA}(\struct{N})$ are the subject of the thesis of Rafal Stefanski \cite{stefthesis}, \cite{DBLP:journals/corr/abs-1907-10504}. The author developed a model of restricted deterministic automata whose languages are recognizable by finitary monoids. In this paper, we shall take another path and generalise the concept of a monoid. If $\catw{Set}(T)$ is the topos under consideration, then by $\catw{Rel}(T)$ we shall denote the category of binary relations in $\catw{Set}(T)$. Category $\catw{Rel}(T)$ equipped with the cartesian product $\times$ and the terminal object $1$ from $\catw{Set}(T)$ forms a monoidal category. By a promonoid in $\catw{Set}(T)$ we shall mean a monoid object in $\catw{Rel}(T)$. Explicitly, a promonoid $\struct{M}$ consists of an object $M$ together with the multiplication relation $\dist{\mu}{M \times M}{M}$ and the unital monomorphism $\mor{\eta}{M_0}{M}$ subject to the usual monoid laws. The category of promonoids and their homomorphisms will be denoted by $\catw{ProMon}(T)$. Because $\catw{Rel}(T)$ has small coproducts inherited from $\catw{Set}(T)$, for every $\Sigma$ there is a free promonoid $\Sigma^*$, which coincides with the free monoid in $\catw{Set}(T)$.

We should also observe that every promonoid has a representation as a monoid, i.e.: every promonoid $\struct{M}$ gives rise to the power monoid $\struct{P(M)}$ by convolution: the unit $1 \rightarrow \Omega^M$ is just the characteristic map of $\eta$, and the multiplication $\Omega^M \times \Omega^M \rightarrow \Omega^M$ is given as the free cocontinous extension of $M \times M \rightarrow \Omega^M$ on each coordinate.


To our surprise, the concept of recognisability by promonoids has not been studied before. Therefore, the next theorem and the following Corollary~\ref{c:pro:regular} has been unknown even in case of the usual nominal sets.

\begin{theorem}[Characterisation of non-deterministic regular languages]\label{t:promonoids}
Let $\mathcal{K}$ be a class of objects closed under binary products. The languages recognized by non-deterministic $\mathcal{K}$-automata coincide with the languages recognized by $\mathcal{K}$-promonoids. 
\end{theorem}
\begin{proof}
Let us observe that for every object $M$ the object $\Omega^{M\times M}$ carries a canonical monoidal structure of binary relations under composition. Because the composition is cocontinuous in both variables, $\Omega^{M\times M}$ is freely generated by its restriction to the singletons, i.e.~by a promonoid $\struct{R_M} = \tuple{M \times M, {\circ}, {=}}$, which we shall call the promonoid of binary relations on $M$. Every promonoid $\struct{M} = \tuple{M, \mu, \eta}$ has a relational representation as a submonoid of the promonoid of $\struct{R_M}$ given by the transposition of its multiplication $\mu$, i.e.~$\dist{\mu^\dag}{M}{M \times M}$ is a homomorphism in the category of promonoids\footnote{One may treat this fact as the generalisation of the Cayley representation for a monoid $\struct{M}$ as a submonoid of the  endo-monoid $M^M$ under functional composition.}.

We claim that if a language $L$ is recognised by a promonoid $\struct{M}$ then it is recognized by promonoid $\struct{R_M}$. Let $\dist{F}{M}{1}$ be a characteristic function of a subobject of $M$. Then $\dist{\eta^\dag \times F}{M \times M}{1}$ is a characteristic function of a subobject of $M \times M$. Moreover, $\mu^\dag \circ (\eta^\dag \times F) = F$ by the definition of the transposition and  neutrality of $\eta$ under $\mu$. Therefore, if there is a relational homomorphism $\dist{\overline{h}}{\Sigma^*}{\struct{M}}$ such that: $L = F \circ \overline{h}$, then $L$ is recognised by homomorphism $\dist{\mu^\dag \circ \overline{h}}{\Sigma^*}{\struct{R_M}}$ with subobject $\eta^\dag \times F$. 

Now, if we define a non-deterministic automaton $\struct{A}$ as: 
\begin{itemize}
\item its object of states is $M$
\item its transition relation is $\dist{\mu \circ (h \times \id{M})}{\Sigma \times M}{M}$
\item its initial states are $\eta$
\item its final states are $F$
\end{itemize}
then, $L(\struct{A})$ is given by:
$$\Sigma^* \overset{\id{\Sigma^*} \times \eta}{\slashedrightarrow} \Sigma^* \times M \overset{\overline{h} \times \id{M}}\slashedrightarrow M \times M \overset{\mu}\slashedrightarrow M \overset{F}\slashedrightarrow 1$$   
because $\mu^\dag \circ \overline{h}$ is a homomorphism as has been shown in the above. But $F \circ \mu \circ (\overline{h} \times \id{M}) \circ (\id{\Sigma^*} \times \eta) = F \circ \mu \circ (\overline{h} \times \eta) = F \circ \overline{h}$ what completes this part of the proof.

In the other direction, let us assume that $L$ is recognized by an automaton $A = \tuple{\mor{s_0}{I}{S}, \mor{s_f}{F}{S}, \dist{\sigma}{\Sigma \times S}{S}}$. Then $L$ is given as the left path on the following diagram:
$$\bfig
\node a1(0, 0)[\Sigma^* \times S]
\node b1(500, 0)[S]
\node a2(0, 300)[\Sigma^*]
\node b2(500, 300)[S \times S]
\node c(900, 0)[1]

\arrow|l|/@{->}/[a2`b2;\overline{\sigma^\dag}]
\arrow|l|/@{->}/[a1`b1;\overline{\sigma}]

\arrow|l|/@{->}/[a2`a1;\id{\Sigma^*} \times \chi^\dag_{s_0}]
\arrow|r|/@{->}/[b2`b1;\chi_{s_0} \times \id{S}]
\arrow|l|/@{->}/[b1`c; \id{S} \times \chi_{s_f}]

\efig$$
The square commutes by the definition of relational composition. Therefore, if we equip $\struct{R_S}$ with ${\chi_{s_0} \times \chi_{s_f}}$, then $\dist{\overline{\sigma^\dag}}{\Sigma^*}{\struct{R_S}}$ will recognize $L$.
\end{proof}

Because $T$-definable objects are closed under binary products, from Theorem~\ref{t:promonoids} we can get the following characterisation of definable non-deterministic languages. 
\begin{corollary}\label{c:pro:regular}
A language can be recognised by a $T$-definable promonoid if and only if it can be recognised by a $T$-definable non-deterministic automaton.
\end{corollary}

\section{Conclusions and further work}\label{s:conclusions}
This paper makes the following contributions. First of all, we show that mathematics can be transferred back and forth between sets with atoms and categories of continuous actions of topological groups (Theorem~\ref{t:transfer} and Theorem~\ref{t:extended:transfer}). Because the topos of continuous actions is much better behaved than the topos of sets with atoms, this allows for a simplification of the mathematical reasoning. For our second contribution, we showed the limits of the classical approach to computability in sets with atoms. It may be inferred from the analysis in \cite{DBLP:conf/lics/BojanczykKL11} that effectiveness of the naive algorithms to the reachability-like problems defined in a decidable complete first order theory is equivalent to $\omega$-categoricity of the theory. Our Theorem~\ref{t:stage} shows that  $\omega$-categoricity of the theory is actually equivalent to the  existence of \emph{any} effective algorithm for reachability-like problems. This leads to our third contribution. We showed how to push forward the concept of algorithms and automata beyond complete first order theories. This requires replacing toposes of continuous actions of topological groups by general classifying toposes for positive existential theories. We have coined a new property of a theory: ``elimination of transitive closures'' and showed that in some aspects it behaves like $\omega$-categoricity for complete first order theories. This includes Theorem~\ref{t:effective:reachability} for the reachability problem, Theorem~\ref{t:effective:emptiness} for the emptiness problem of an automaton, and Nihil-Nerode like Theorem~\ref{t:definable:mn}, which is central for studying behaviours of deterministic automata. For our forth contribution, we established a general correspondence between languages of non-deterministic automata and relational monoids in Theorem~\ref{t:promonoids}. This correspondence has not been known before even for very restricted cases (like nominal sets). The meta-contribution of this paper is in showing that many concepts incarnate in different areas of mathematics; by linking these incarnations together we can simplify our thoughts and proofs. For example, the connection between coherent groups (defined by Blass and Scedrov in '80s to characterise coherent toposes of continuous actions of topological gorups) and Roelcke precompact groups (defined by Roelcke in '70s to characterise topological dynamics) has not been observed before. Similarly, many of the theorems from \cite{bojanbook} with advanced proofs are easy consequences of the facts from the theory of classifying toposes and the connection established in this paper (compare the proof of Theorem~5.1 in \cite{bojanbook} with our Theorem~\ref{t:effective:reachability}).

For further work we shall study other concepts and algorithms definable in positive existential theories, e.g.:~constraint satisfaction problems with definable sets of constraints, definable pushdown automata, definable Turing machines, etc. Our recent paper \cite{ameMRP} shows that carrying over some of these results to non-Boolean classifying toposes is a challenging task.

\vskip 1em
\begin{acknnowledgments}
The authors would like to thank Alex Kruckman for his help in formalizing the proof of Theorem~\ref{t:stage}.

This research was supported by the National Science Centre, Poland, under projects 2018/28/C/ST6/00417.
\end{acknnowledgments}

\bibliographystyle{apalike}
\bibliography{bibl}


\end{document}